\documentclass[aps,physrev,reprint,]{revtex4-2}

\usepackage{amsmath, amsthm, amsfonts, graphicx, subfigure, braket, dcolumn,tcolorbox}
\usepackage[colorlinks,citecolor=blue,linkcolor=blue, urlcolor=blue, anchorcolor=blue]{hyperref}

\usepackage{algorithm}
\usepackage[noend]{algpseudocode}
\usepackage{tikz}

\newtheorem{lemma}{Lemma}
\newtheorem{theorem}{Theorem}

\newcommand{\arccot}{\mathrm{arccot}\,}

\begin{document}

\title{Robust Quantum Walk Search Without Knowing the Number of Marked Vertices}

\author{Yongzhen Xu}
\author{Delong Zhang}
\author{Lvzhou Li}\thanks{Corresponding author.\\ lilvzh@mail.sysu.edu.cn}
	\affiliation{Institute of Quantum Computing and Computer  Theory, School of Computer Science and
Engineering, Sun Yat-sen University, Guangzhou 510006, China}

\begin{abstract}
 There has been a very large body of research  on  searching a marked vertex on a  graph based on quantum walks, and Grover's algorithm can be regarded as a quantum walk-based search algorithm on a special graph.  However,   the existing quantum walk-based search algorithms suffer severely from the souffl\'{e}  problem which mainly  means that the success probability of finding a marked vertex could shrink dramatically even to zero when the number of search steps is greater than the right one, thus heavily reducing the robustness and practicability of the algorithm. Surprisingly, while the souffl\'{e}  problem of Grover's algorithm has attracted enough attention, how to address this problem for general quantum walk-based search algorithms is missing in the literature.  Here we initiate the study of  overcoming the  souffl\'{e}  problem for quantum walk-based search algorithms by presenting a new     quantum walk-based search framework that achieves  robustness without sacrificing the quantum speedup. In this framework,  for any adjustable parameter $\epsilon$, the quantum algorithm  can find a marked vertex on an $N$-vertex {\it complete bipartite graph} with probability at least $ 1-\epsilon$, whenever the number of search steps  $h$ satisfies  $h \geq \ln(\frac{2}{\sqrt{\epsilon}})\sqrt{N} + 1$. Note that the algorithm need not  know  the exact number of marked vertices.  Consequently, we obtain quantum search algorithms with stronger robustness and practicability.
\end{abstract}

\maketitle
\section{Introduction}
Quantum walks, the analogues of classical random walks in the quantum realm, were first introduced by Aharonov, Davidovich and Zagury~\cite{AharonovDZ93PhysRevA} in 1993.
In the last nearly thirty years, much progress about quantum walks has been made from theory to experiments.
Quantum walks have become  a basic tool for designing quantum algorithms to settle a series of problems such as element distinctness~\cite{Ambainis07}, triangle finding~\cite{MagniezSS07}, quantum backtracking~\cite{Montanaro18},  and so on~\cite{BuhrmanS06, MagniezN07,ChildsSV07, JefferyKM13,BelovsCJKM13}. Furthermore, they are a universal model of quantum computation~\cite{Childs09PRL,LovettCETK10PRA}.
In the aspect of experiment study, various hardware platforms have been used to demonstrate  results of quantum walks, e.g.,~\cite{PhysRevLett2010,science20101193515,PhysRevLett2018,science2021}. There are two types of quantum walks:  discrete-time quantum walks~\cite{Meyer93,MEYER1996337,Watrous01,AmbainisBNVW01,AmbainisKV01} and continuous-time quantum walks~\cite{FarhiG98PRA,ChildsFG02}. In this paper, we are concerned with the discrete-time model.

A central topic in quantum walk-based algorithms is
to develop efficient quantum algorithms for searching a marked vertex on a graph. This idea was initially proposed by Shenvi, Kempe, and Whaley~\cite{Shenvi2003PhysRevA} in 2003 who constructed a quantum walk search algorithm on the Boolean hypercube for finding a marked item in a dataset. Later, Ambainis, Kempe and Rivosh~\cite{AmbainisKR05} proposed search algorithms based on quantum walks on $d$-dimensional lattices $(d \geq 2)$. A major breakthrough was made in 2004: Ambainis~\cite{Ambainis07} obtained the optimal query complexity of the element distinctness problem by employing quantum walk search on  Johnson graphs. In 2004, Szegedy~\cite{Szegedy04} studied the general theory of quantum walk search algorithms   from the point of view of Markov chains. In this direction, a series of work~\cite{MagniezNRS11, KroviMOR16, AmbainisGJK20, ApersGJ21} was put  forward for searching a marked state in different Markov chains using phase estimation, interpolated quantum walks and quantum fast-forwarding.

Grover's algorithm can be regarded as a quantum walk search algorithm on a complete graph with a
self-loop on every vertex~\cite{AmbainisKR05}. As pointed out by  Brassard \cite{Science1997},  Grover's quantum searching technique suffers from the souffl\'{e} problem \cite{soffle}. As a result, if  the exact number of marked items is not known in advance, then one does not know when to stop the search iteration. Even if  the number is known, the success probability of the algorithm could shrink dramatically when the number of query steps is greater than the right one, as shown in Fig.~\ref{fig001}(a).
Two strategies are often used to deal with the unknown number of solutions. One method is the exponential search algorithm~\cite{1998Boyer}
in which the number of iterations increases slowly but exponentially. Another strategy is to employ quantum counting~\cite{2002Brassard} to estimate the number of marked items. However, they are still essentially oscillatory Grover search and fail to completely solve the souffl\'{e} problem.
From a search perspective, the success probability of getting a marked item should not shrink (at least not shrink dramatically) as the number of search steps increases.

In order to overcome the souffl\'{e} problem, Grover \cite{Grover05PRL} proposed a fixed-point quantum search algorithm that  converges monotonically to the
target, i.e., avoid overcooking by always amplifying the marked items (as shown in Fig.~\ref{fig001}(b)). Yet, a price
paid for this monotonicity  is that the quadratic
speedup of the original quantum search is lost. 
In 2014, Yoder, Low and Chuang~\cite{YoderLC14} presented a new  quantum search algorithm  that achieves both goals---the search cannot be overcooked and also achieves
optimal time scaling, a quadratic speedup over classical
unordered search. This algorithm
does not require  that the error monotonically
improves, but ensures that the error
becomes bounded by a tunable parameter $\epsilon$, as shown in Fig.~\ref{fig001}(c). Thus, this leads to a more robust  quantum search algorithm. In addition, the fixed-point quantum search was discussed from an information perspective by Cafaro \cite{CAFARO2017154} and from the view of analog quantum search with suitable Hamiltonians specifying time-dependent two-level quantum systems by Cafaro and Alsing \cite{cafaro2019}.

\begin{figure}[htp]
    \centering
    \includegraphics[width=7.5cm]{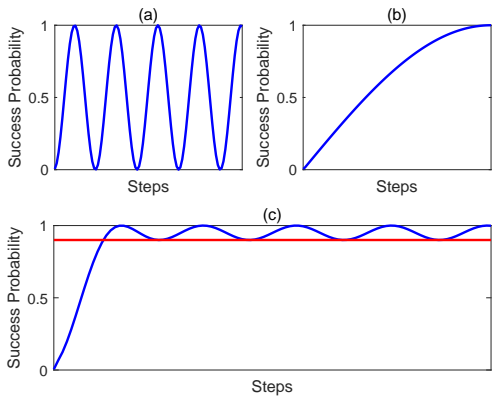}\\
    \caption{The success probability of finding a marked item as a function
of steps in three search models. (a) Grover-type oscillatory search. (b) Fixed-point Search. (c) Robust search.   }
    \label{fig001}
\end{figure}
Currently, various quantum walk search algorithms also suffer from the souffl\'{e}  problem. To the best of our knowledge, there has been no work considering how to avoid the soffl\'{e} problem from the perspective of quantum walk search. Actually,  the problem becomes more challenging than in Grover's quantum search.
There are at least three  reasons for the difficulty of addressing the souffl\'{e}  problem for quantum walk search algorithms. Firstly,
the search space is more complicated because of the diversity of topological structure of graphs.  Secondly, more operations are involved in quantum walk search. Thirdly, it is generally difficult to get an analytical expression for the success probability and  to analyze the computational complexity theoretically of a  quantum walk search algorithm.

\subsection{Our contributions}
This paper   considers for the first time how to address the souffl\'{e}  problem confronted by quantum walk search algorithms. We initiate the first step towards this direction   by designing a robust quantum walk search algorithm  on complete bipartite graphs (NOT complete graphs). Note that this kind of  graphs was extensively studied in quantum walk search algorithms~\cite{ReitznerHFB09,RhodesT19,2014CH,2022QMWXWX}. The robustness feature of our algorithm ensures that for an $N$-vertex complete bipartite graph with marked vertices  but without knowing the number of marked vertices, if the number of search steps $h$ satisfies  $h \geq \ln(\frac{2}{\sqrt{\epsilon}})\sqrt{N}+1$, then  the algorithm will output a marked vertex with probability at least $ 1-\epsilon$ for any  given $\epsilon\in
(0,1]$ (the formal statement can be found in Theorem \ref{oneOddTheorem}).  Thus, the algorithm both avoids overcooking and keeps  quadratic speedup over classical ones. Also note that our algorithm need not  know the number of target vertices.

In order to obtain the above result,  some nontrivial technical treatments are required.
 (1)  First, Compared to Grover's algorithm, the coined quantum walk search framework has two subsystems and three operations. Thus, what operations should be adjusted  to create a robust version is not obvious, and we show that a model with two parameterized operations  is sufficient.
(2) Second,  while  one needs only consider essentially a  two-dimensional state space for Grover's algorithm,  higher dimensional state spaces are involved in quantum walk search and thus it is even not easy to track the state of the quantum system. Luckily,
we reveal some crucial observations (especially Lemma \ref{lemmasabs}) to  simplify the expression of the final state. It is worth noting that these observations  are specific to quantum walks and are not seen in the robust version of Grover's algorithm  \cite{YoderLC14}.
    We think that these technical treatments may inspire the analysis of robust quantum walk search on other general graphs.
\begin{theorem}\label{oneOddTheorem}
Given    an N-vertex complete bipartite graph with
marked vertices but without knowing the number of marked vertices, there exists a quantum walk-based algorithm such  that if the number of search steps $h$ satisfies  $ h \geq \ln(\frac{2}{\sqrt{\epsilon}})\max(\sqrt{N_l}, \sqrt{N_r}) +1$, then the algorithm  will output a marked vertex with probability at least $ 1-\epsilon$ for any given $\epsilon \in (0,1]$, where $N_l$ ($N_r$)  is the number of the left (right) vertices in the complete bipartite graph.
\label{thm-1}
\end{theorem}

The relationship among  the theorems and technical lemmas obtained in this paper are depicted  in  Fig.~\ref{figTheorem}.  Theorem~\ref{oneOddTheorem} states the main result of this paper  which  comes from Theorems~\ref{oneSetTheorem} and \ref{twoSetTheorem} which corresponds respectively to the two cases: the marked vertices are in  one side and in two sides of a complete bipartite graph. Furthermore, Lemma  \ref{lemma2} ( Lemma  \ref{lemma3}) is crucial for proving  Theorem~\ref{oneSetTheorem} (Theorem \ref{twoSetTheorem}), with proofs given in  Section \ref{method} and Appendixes \ref{AppendixA}, \ref{AppendixB}, \ref{AppendixC} and \ref{AppendixD}.

\begin{figure}[htp]
    \centering
     \includegraphics{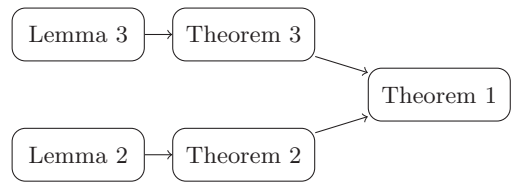}\\
    \caption{The relationship among the theorems and technical lemmas.}
    \label{figTheorem}
\end{figure}

\section{Preliminaries}
\medskip
\noindent\textbf{Graph Notation.}
Let $G=(V,E)$ be an undirected, unweighted graph with $N=|V|$ vertices and $m=|E|$ edges. An edge between $u$ and $v$ is denoted by $(u,v)$. For $u \in V$, $deg(u)=\{v: (u,v)\in E\}$ denotes the set of neighbors of $u$, and the degree of $u$ is denoted as $d_u=|deg(u)|$. A bipartite graph is represented as $G=(V=\{V_l \cup V_r\}, E$), where $V_l $ ($V_r $) denotes the set of vertices in the left (right) side, with $V_l \cap V_r=\emptyset$. We use $N_l$ and $N_r$ to denote the number of left and right vertices, respectively. The number of the marked vertices in the left (right) side is $n_l$ ($n_r$).
A complete bipartite graph is a  bipartite graph where every vertex in the left side is connected to every vertex in the right side. For example, a complete bipartite graph in Fig.~\ref{fig:1a} contains $6$ vertices in the left side and $4$ vertices in the right side. 


\begin{figure}[htp]
\centering
\subfigure[] { \label{fig:1a}
\includegraphics[width=0.13\textwidth]{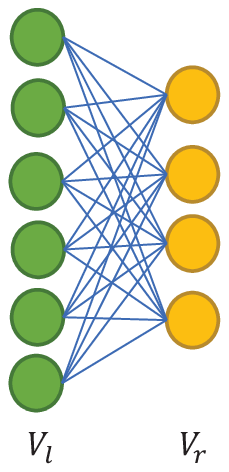}
}
\subfigure[] { \label{fig:1b}
\includegraphics[width=0.13\textwidth]{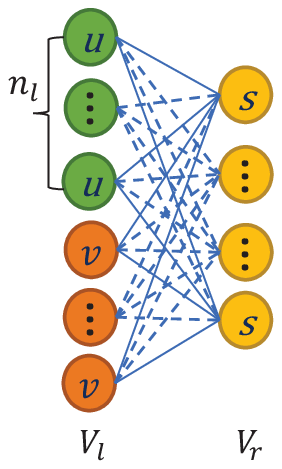}
}
\subfigure[] { \label{fig:1c}
\includegraphics[width=0.13\textwidth]{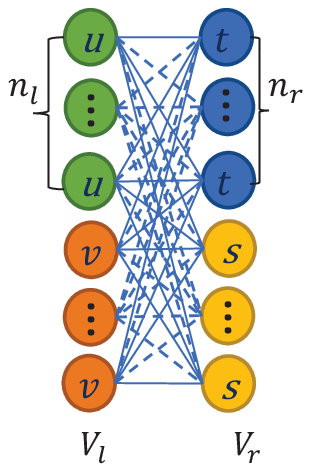}
}
\caption{(a)  A complete bipartite graph G with 10 vertices. (b) An N-vertex complete bipartite graph with the marked vertices denoted by $u$ in the left, the unmarked vertices denoted by $v$ in the left and $s$ in the right. (c)  An N-vertex complete bipartite graph with the marked vertices denoted by $u$ in the left and $t$ in the right, the unmarked vertices denoted by $v$ in the left and $s$ in the right.}
\label{fig:1}
\end{figure}

\medskip
\noindent\textbf{Coined Quantum Walk.}
In this model, the walker's Hilbert space associated with an $N$-vertex graph $G=(V,E)$ is $\mathbb{H}^{N^2}= span\{\ket{uv}, u,v \in V\}$,
where  $u$ is the position and $v$ is the coin value representing one  neighbor of $u$. The evolution operator of the coined quantum walk at each step is $ U_{\text{walk}}=SC,$
where the flip-flop shift operator $S$ is  defined as $S\ket{uv}=\ket{vu},$
and the coin operator is
$C=\sum_u\ket{u}\bra{u}\otimes C_u.$
 The Grover diffusion coin operator $C_u$ often used is $C_u=2\ket{s_u}\bra{s_u}-I,$
where $\ket{s_u}= \frac{1}{\sqrt{d_u}} \sum_{v\in deg(u)}\ket{v}.$

Given the initial state $\ket{\Psi_0}$ ,  the walker's state after $h$ steps is
$\ket{\Psi_h}=U_{\text{walk}}^h\ket{\Psi_0}$.

\medskip
\noindent\textbf{Quantum Walk Search.}
In the quantum walk search framework, to find a marked vertex in a graph, the query oracle $Q$ is given by
\begin{equation*}
    Q\ket{uv}=
    \begin{cases}
    -\ket{uv}& \text{if $u$ is marked},\\
    \quad \ket{uv}& \text{if $u$ is not marked.}
    \end{cases}
\end{equation*}
The evolution operator corresponding to one step of the quantum walk search  is $U=SCQ.$

Given the initial state $\ket{\Psi_0}$, the walker's state after $t$ steps is
$\ket{\Psi_t}=U^t\ket{\Psi_0}$. Finally,  the first register is measured and the measurement result is output.

\medskip
\noindent\textbf{Chebyshev polynomial.}
The Chebyshev polynomials of the first kind $T_n(x)$ are defined by initial values $T_{0}(x)=1$, $T_{1}(x)=x$, and for an
integer $n \geq 2,$
\begin{equation*}
    T_n(x)=2xT_{n-1}(x)-T_{n-2}(x).
\end{equation*}
The trigonometric identity $T_{n}(x)=\cos(n\arccos(x))$ is well known.

A result of one Quasi-Chebyshev polynomial implied  in \cite{YoderLC14} is stated in the following lemma.

\begin{lemma}\label{lemma1}Let $x=\cos(\theta)$ for $\theta \in [0,2\pi]$. Let $h \geq 3$ be an odd integer. One Quasi-Chebyshev polynomial $a_k(x)$ is defined by initial values $a_0(x)=1$, $a_1(x)=x$,
and for $k=2,\dots,h,$
\begin{equation*}
   a_k(x)=x(1+e^{-i(\zeta_k-\zeta_{k-1})})a_{k-1}(x)-e^{-i(\zeta_k-\zeta_{k-1})}a_{k-2}(x).
\end{equation*}
When $\zeta_{k+1}-\zeta_k=(-1)^k\pi-2 \arccot\left(\tan(k\pi/h)\sqrt{1-\gamma^2}\right)$ for $k=1,\dots,h-1$, where $\gamma = \frac{1}{cos(\frac{1}{h}\arccos(\frac{1}{\sqrt{\epsilon}}))}$ with $\epsilon \in(0,1]$, we have $a_h(x)=\frac{T_h(x/\gamma)}{T_h(1/\gamma)}$ with $T_h(1/\gamma)=1/\sqrt{\epsilon}$.
\end{lemma}

\section{Robust quantum walk search on complete bipartite graphs}
As mentioned before, the already existing quantum walk search algorithms suffer from the  souffl\'{e} problem. Thus, this work devotes to  addressing this problem  by considering the case of  searching a marked vertex in a complete bipartite graph. For that, first
 the coin operator $C$ and the query oracle $Q$ have  to be adjusted, but the flip-flop shift operator $S$  can remain unchanged.
The new evolution operator of one search step is
\begin{equation}\label{Ualphabeta}
    U(\alpha,\beta)=SC(\alpha)Q(\beta),
\end{equation}
where the coin operator $C$ is changed to
\begin{equation}\nonumber
    C(\alpha) =\sum_u\ket{u}\bra{u} \otimes[(1-e^{-i\alpha})\ket{s_u}\bra{s_u}-I],
\end{equation}
and the query oracle $Q$ is replace by
\begin{equation*}
    Q(\beta)\ket{uv}=
    \begin{cases}
    e^{i\beta}\ket{uv}& \text{if $u$ is marked},\\
     \ket{uv}& \text{if $u$ is not marked.}
    \end{cases}
\end{equation*}
When $\alpha=\beta=\pm\pi$, this model becomes the original quantum walk search~\cite{AmbainisKR05,RhodesT19}.

 The algorithm of search on a complete bipartite graph is given in Algorithm \ref{algorithm}.
In the input phase,
according to  the information of marked vertices and a tunable parameter $\epsilon$, the number of search steps $h$ is required to satisfy

\begin{widetext}
\begin{equation}\label{steph}
    h\geq
    \begin{cases}
    \ln(\frac{2}{\sqrt{\epsilon}})\sqrt{\frac{N_l}{n_l}}+1 &\text{marked vertices in one side with $n_l \geq 1,n_r=0$,}\\
   \ln(\frac{2}{\sqrt{\epsilon}})\max(\sqrt{\frac{N_l}{n_l}},\sqrt{\frac{N_r}{n_r}}) +1 &\text{marked vertices in two sides with $n_l \geq 1,n_r\geq 1$,}\\
    \ln(\frac{2}{\sqrt{\epsilon}})\max(\sqrt{N_l}, \sqrt{N_r}) +1 &\text{without knowing the number and any arrangement of marked vertices.}\\
    \end{cases}
\end{equation}
The  parameters $\alpha, \beta$ are given by

\textbf{Case 1: $h$ is odd,}

\begin{equation}\label{alphabetaOddh}
    \alpha_k=
    \begin{cases}
        -\beta_{h+2-k}=2\arccot(\tan(\frac{k\pi}{h})\sqrt{1-\gamma^2}) &
                       \text{ $k=2,4,\dots,h-1$},\\
    -\beta_{h-k}=2\arccot(\tan(\frac{(k-1)\pi}{h})\sqrt{1-\gamma^2})&
                  \text{ $k=3,5,\dots,h,$}\\
\text{$\alpha_1$ and $\beta_h$ can be any value,}
    \end{cases}
\end{equation}
where $\gamma^{-1}=\cos(\frac{1}{h}\arccos(\frac{1}{\sqrt{\epsilon}}))$.

\textbf{Case 2: $h$ is even,}
\begin{equation}\label{alphabetaEvenh}
\alpha_k=-\beta_{h+1-k}=
\begin{cases}
2\arccot(\tan(\frac{k\pi}{h+1})\sqrt{1-\gamma_1^2}) &\text{$k=2,4,\dots,h$}\\
2\arccot(\tan(\frac{(k-1)\pi}{h-1})\sqrt{1-\gamma_2^2}) &\text{$k=3,5,\dots,h-1$,}\\
\text{$\alpha_1$ and $\beta_h$ can be any value,}
\end{cases}
\end{equation}
where $\gamma^{-1}_1=\cos(\frac{1}{h+1}\arccos(\frac{1}{\sqrt{\epsilon}}))$ and $\gamma^{-1}_2=\cos(\frac{1}{h-1}\arccos(\frac{1}{\sqrt{\epsilon}}))$.
\end{widetext}

\begin{algorithm}[H]
\caption{Robust quantum walk search}
\label{algorithm}
\begin{description}
    \item[Inputs] an N-vertex complete bipartite graph with
marked vertices, $\epsilon \in (0,1]$, and the number of search steps $h$.
    \item[Outputs] a marked vertex  $x_0$  (if $h$ satisfies Eq.~\eqref{steph}, it outputs a marked vertex with probability at least $ 1-\epsilon$).
    \item[Procedure]
\end{description}
    \begin{enumerate}
        \item Prepare the initial state $\ket{\Psi_0}= \frac{1}{\sqrt{2N_l N_r}}(\sum_u \ket{u}\otimes\sum_{v\in deg(u)}\ket{v})$.
        \item Apply $U(\alpha_1,\beta_1),..., U(\alpha_h,\beta_h)$ in turn, where the parameters $\alpha_i, \beta_i$ are determined by Eqs.~\eqref{alphabetaOddh} and~\eqref{alphabetaEvenh}.
        \item Measure the first register in the computational basis. If the result vertex is not marked, then the second register is measured. Output the measurement result.
    \end{enumerate}
\end{algorithm}
In the first step, the initial state $\ket{\Psi_0}= \frac{1}{\sqrt{2N_l N_r}}(\sum_u \ket{u}\otimes\sum_{v\in deg(u)}\ket{v})$ is prepared. Then, $U(\alpha_1,\beta_1),..., U(\alpha_h,\beta_h)$  with  appropriate parameters are applied to $\ket{\Psi_0}$ in turn. The whole  operator that performs $h$ steps is $\varGamma_h =U(\alpha_h,\beta_h)...U(\alpha_1,\beta_1)$.
The walker's state after $h$ steps is
$\ket{\Psi_h}=\varGamma_h\ket{\Psi_0}.$
Finally, the two registers are measured. Note that in the previous work generally only the first register is measured, whereas here we measure the two registers. This will double the success probability for our problems as shown later. The success probability of getting a marked vertex is
\begin{equation*}
    P_h = \sum_{ \text{ $u$ or $v$ is marked}}|\bra{uv}\varGamma_h\ket{\Psi_0}|^2.
\end{equation*}

Fig.~\ref{fig:3} illustrates the success probability of finding a marked item as a function of steps in Algorithm~\ref{algorithm}  and the one with $\alpha=\beta=\pm\pi$ in Eq.~\eqref{Ualphabeta}. They show that Algorithm~\ref{algorithm} is a robust search model.

\begin{figure*}[htp]
\centering
\subfigure[] { \label{fig:3a}
\includegraphics[width=0.45\textwidth]{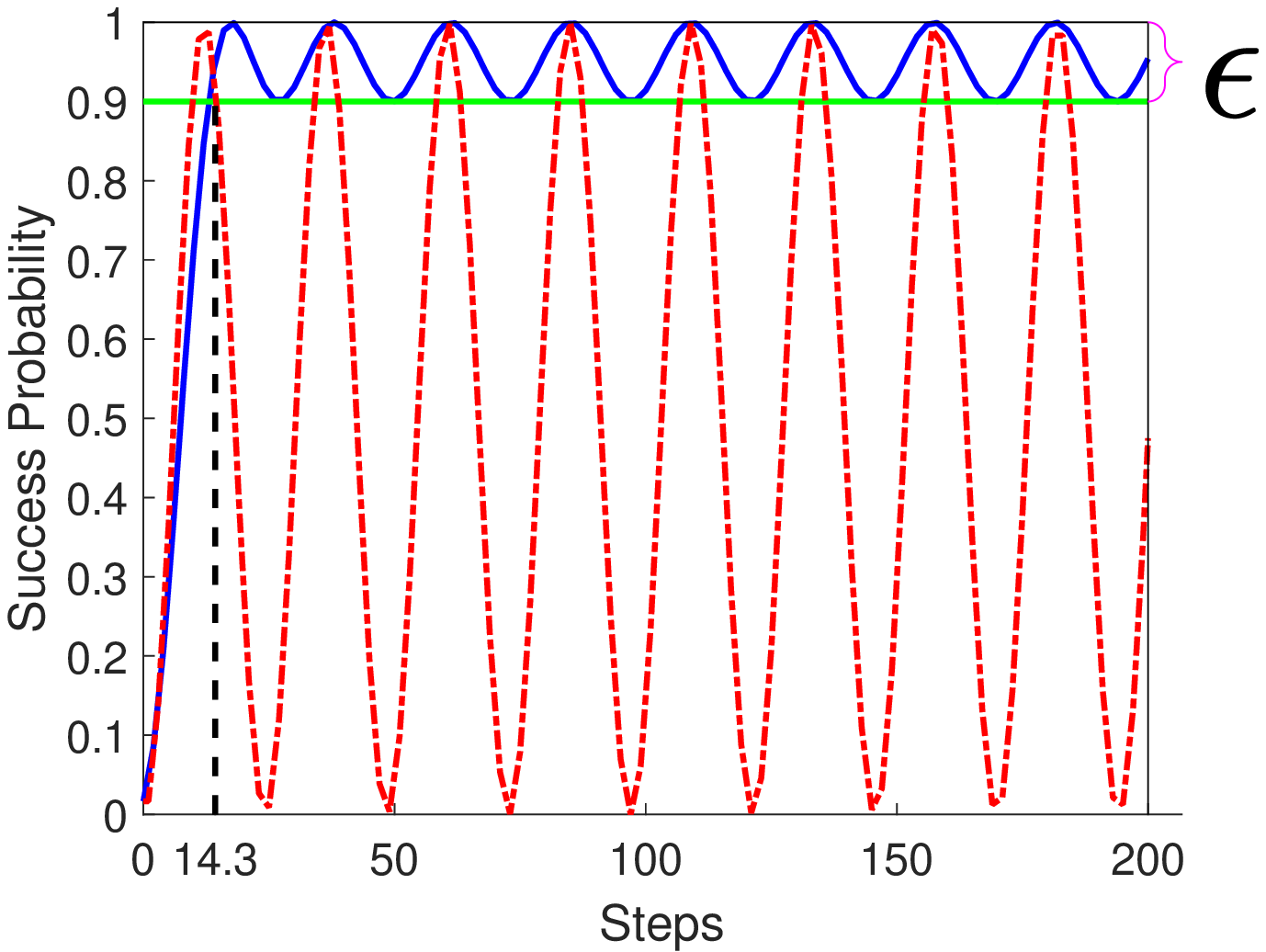}
}
\subfigure[] { \label{fig:3b}
\includegraphics[width=0.45\textwidth]{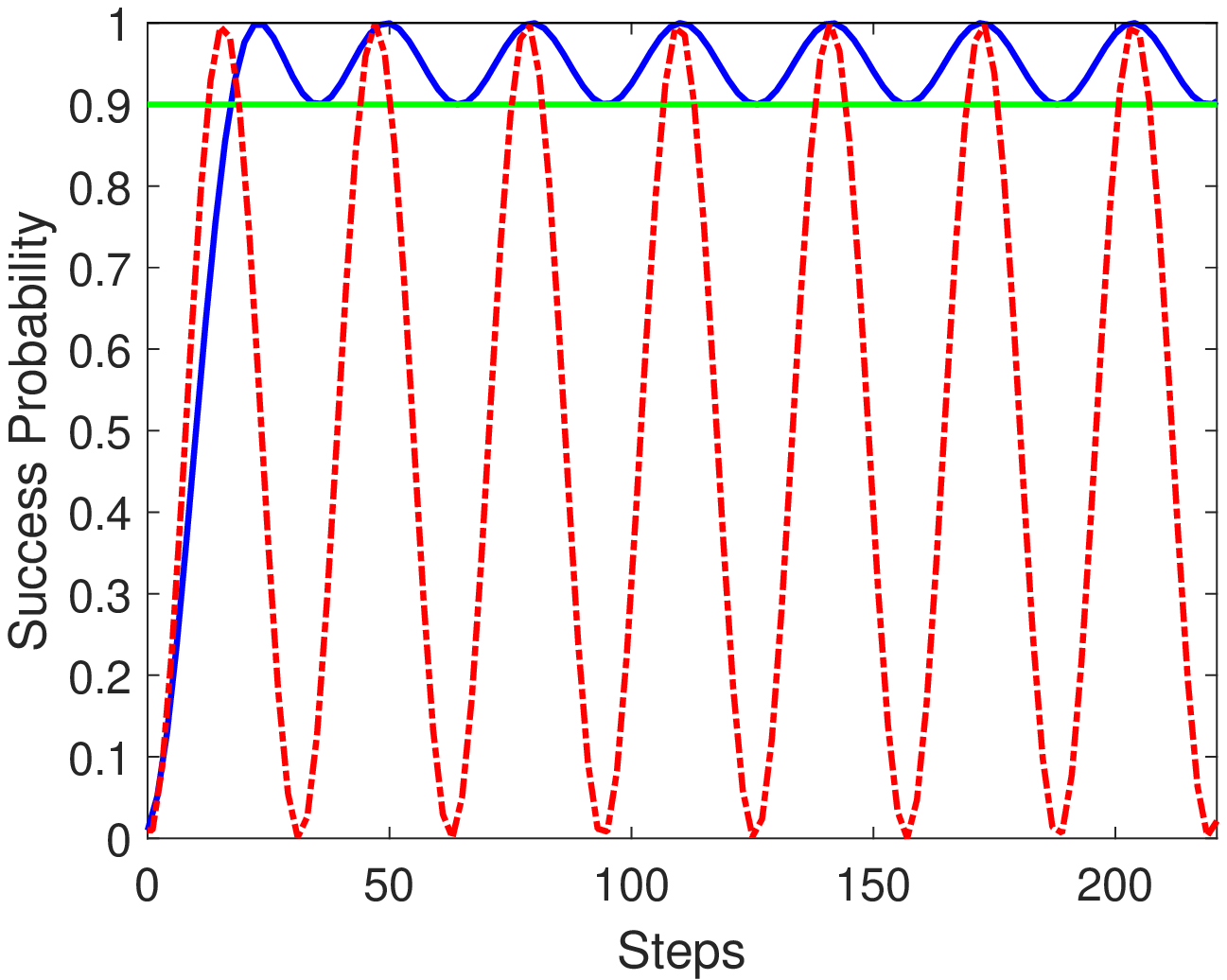}
}
\subfigure[] { \label{fig:3c}
\includegraphics[width=0.45\textwidth]{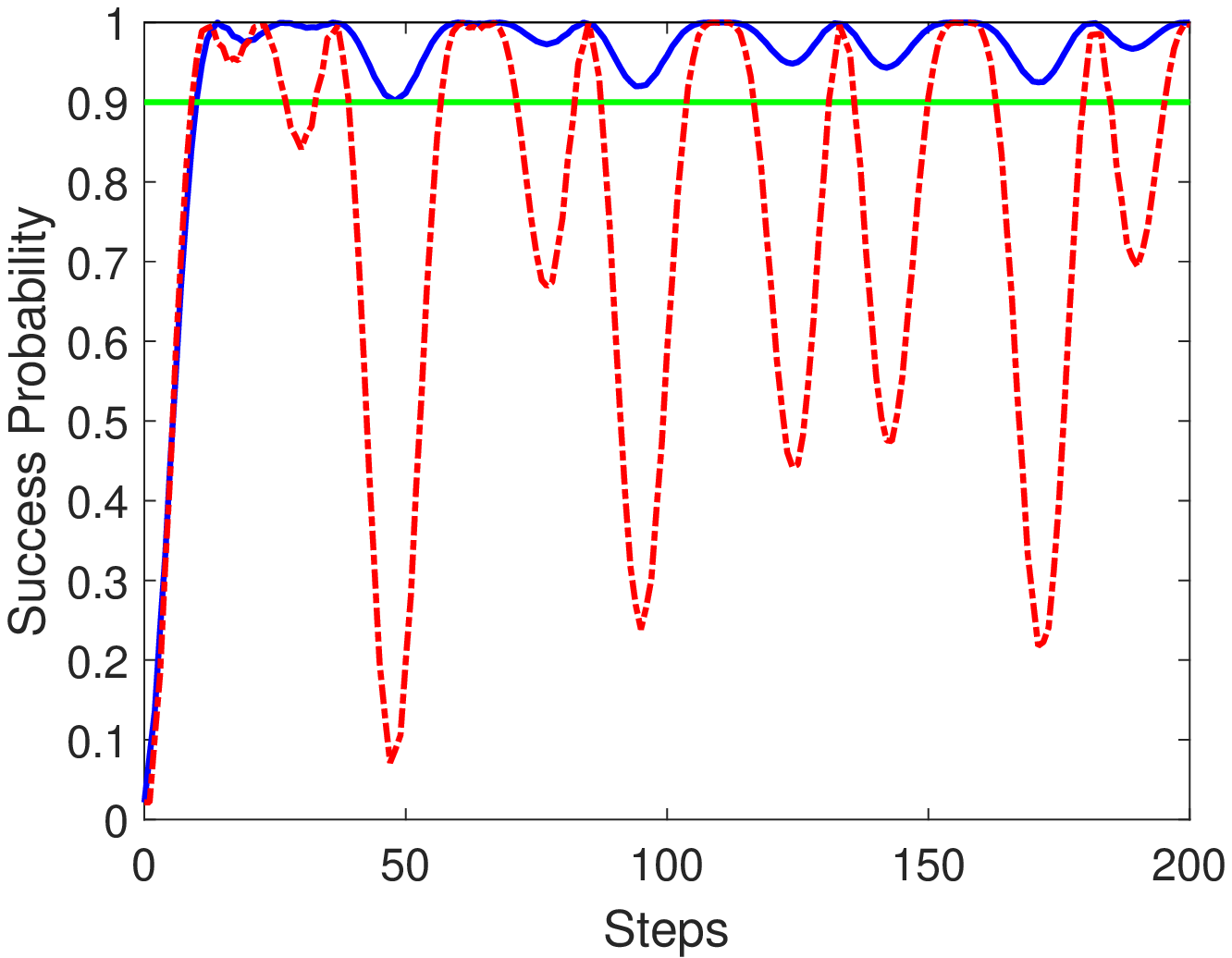}
}
\subfigure[] { \label{fig:3d}
\includegraphics[width=0.45\textwidth]{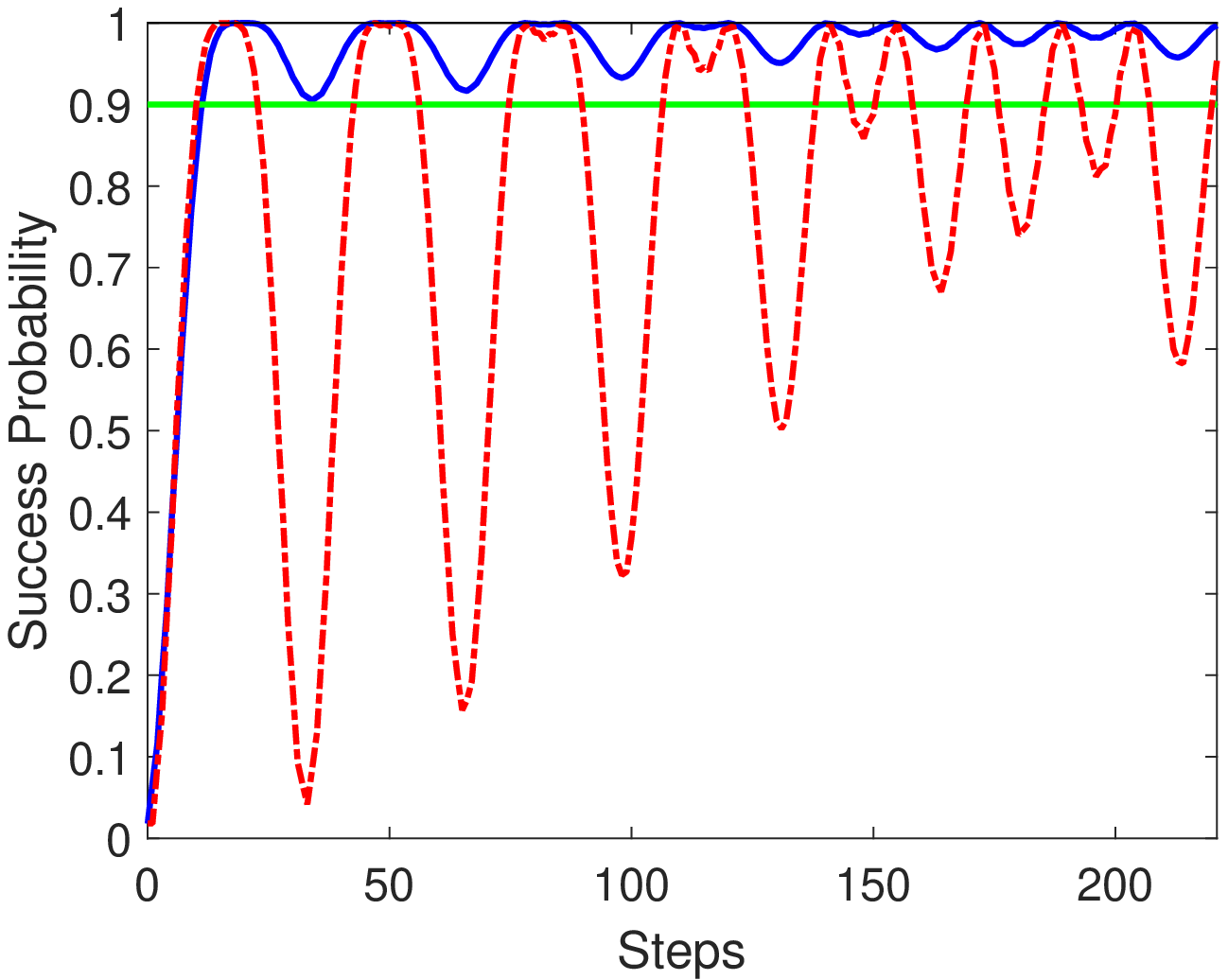}
}
\caption{The success probability of finding a marked item as a function of steps in Algorithm~\ref{algorithm}  (blue solid curve) and the one with $\alpha=\beta=\pm\pi$ in Eq.~\eqref{Ualphabeta} (red dash-dot curve). The green horizontal solid line indicates that the success probability is greater than or equal to $0.9$. (a) $N_l=600, n_l=10,N_r=1000, n_r=0.$ The success probability $P_h  \geq 0.9 $ when $h \geq \ln(\frac{2}{\sqrt{\epsilon}})\sqrt{\frac{N_l}{n_l}}\approx 14.3$ with $\epsilon=0.1$. (b) $N_l=1000, n_l=10,N_r=600, n_r=0.$  (c) $N_l=600, n_l=10,N_r=1000, n_r=5.$ (d) $N_l=1000, n_l=10,N_r=600, n_r=5.$}
\label{fig:3}
\end{figure*}

\medskip
\noindent
\begin{proof}[ \bf  Proof of Theorem \ref{oneOddTheorem}]
According to the arrangement of marked vertices, two cases
are discussed: they are in one side and two sides, as shown in Fig.~\ref{fig:1b} and Fig.~\ref{fig:1c}, respectively.
  \begin{enumerate}
\item [(i)] In the first case, without loss of generality, suppose that  all the marked vertices lie in the left side $V_l$. Then by  Theorem \ref{oneSetTheorem}, if $h\geq \ln(\frac{2}{\sqrt{\epsilon}})\sqrt{\frac{N_l}{n_l}}+1 $,  then Algorithm 1 will  output a marked vertex with probability at least $ 1-\epsilon$.
\item [(ii)] In the second case, by Theorem \ref{twoSetTheorem}, if $h\geq \ln(\frac{2}{\sqrt{\epsilon}})\max(\sqrt{\frac{N_l}{n_l}},\sqrt{\frac{N_r}{n_r}}) +1$,  then Algorithm 1 will  output a marked vertex with probability at least $ 1-\epsilon^2$.

 \end{enumerate}

Note that the parameters  $\alpha_i, \beta_i$ will be assigned with the same values in the  two cases (this can be seen from the proof of Theorems \ref{oneSetTheorem} and \ref{twoSetTheorem}). Thus, we need not know  the arrangement of marked vertices.
In addition, in the two cases, both $\sqrt{\frac{N_l}{n_l}}$ and $\max(\sqrt{\frac{N_l}{n_l}},\sqrt{\frac{N_r}{n_r}})$ have the same upper bound $\max(\sqrt{N_l}, \sqrt{N_r})$. Therefore, if $h \geq \ln(\frac{2}{\sqrt{\epsilon}})\max(\sqrt{N_l}, \sqrt{N_r}) +1$, then
Algorithm~\ref{algorithm} will output a marked vertex with probability at least $1-\epsilon$.
This completes the proof.
\end{proof}

\begin{theorem}\label{oneSetTheorem}
In Algorithm~\ref{algorithm}, suppose that all the marked vertices are in the left side. There exists a sequence of parameters $\alpha_i, \beta_i$, such that if $h\geq \ln(\frac{2}{\sqrt{\epsilon}})\sqrt{\frac{N_l}{n_l}}+1 $,  then the algorithm will output a marked vertex with probability at least $1-\epsilon$, where $N_l$ is the  number of left vertices and $n_l$ is the total number of marked vertices.
\end{theorem}
\begin{proof}

According to Lemma \ref{lemma2},
when $h$ is odd, in order to ensure  $P_h \geq 1-\epsilon$,  it suffices to  satisfy $|\cos(\frac{1}{h}\arccos(\frac{1}{\sqrt{\epsilon}}))\sqrt{1-\frac{n_l}{N_l}}| \leq 1, $ that is,
\begin{align}\label{key}
\frac{n_l}{N_l} &\geq 1-\cos^{-2}(\frac{1}{h}\arccos(\frac{1}{\sqrt{\epsilon}})).
\end{align}
Note that the following functions in (see for instance Ref.\cite{abramowitz1964handbook})
will be used:
\begin{align*}
    &\arccos (z) = \frac{1}{i}\ln (z+\sqrt{z^2-1}), \tan(i z) = i \tanh(z),\\
    &\tanh(x)=\frac{e^x-e^{-x}}{e^x+e^{-x}},
\end{align*}
where $\ln(\cdot)$ is the natural logarithm function,  $i$ denotes the imaginary number, $x$ is a real number and $z$ is a complex number.  Now we have
\begin{align*}
1-&\cos^{-2}(\frac{1}{h}\arccos(\frac{1}{\sqrt{\epsilon}}))\\
                &=  -\tan^2(\frac{1}{h}\arccos(\frac{1}{\sqrt{\epsilon}}))\\
                &=-\tan^2(\frac{1}{h}\frac{1}{i}\ln (\frac{1}{\sqrt{\epsilon}}+\sqrt{(\frac{1}{\sqrt{\epsilon}})^2-1}))\\
                &=-\tan^2(i \frac{1}{h}\ln (\frac{1}{\sqrt{\epsilon}}+\sqrt{(\frac{1}{\sqrt{\epsilon}})^2-1}))\\
                &=\tanh^2(\frac{1}{h}\ln (\frac{1}{\sqrt{\epsilon}}+\sqrt{(\frac{1}{\sqrt{\epsilon}})^2-1}))\\
                &<\tanh^2(\frac{1}{h}\ln (\frac{2}{\sqrt{\epsilon}}))\\ &<\left(\frac{\ln(2/\sqrt{\epsilon})}{h}\right)^2,
\end{align*}
where the last inequality follows from $x\geq \tanh(x)$ for $x \geq 0.$ Thus, in order to ensure the inequality (\ref{key}), it suffices to set $\frac{n_l}{N_l} \geq \left(\frac{\ln(2/\sqrt{\epsilon})}{h}\right)^2$, which leads to $h\geq\ln(\frac{2}{\sqrt{\epsilon}})\sqrt{\frac{N_l}{n_l}}$.

Similarly, when $h$ is even,  in order to ensure $P_h \geq 1-\epsilon$, it  suffices to  satisfy
$ |\cos(\frac{1}{h+1}\arccos(\frac{1}{\sqrt{\epsilon}}))\sqrt{1-\frac{n_l}{N_l}}|\leq 1, $
and
$ |\cos(\frac{1}{h-1}\arccos(\frac{1}{\sqrt{\epsilon}}))\sqrt{1-\frac{n_l}{N_l}})|\leq 1,$
which implies
$h+1\geq\ln(\frac{2}{\sqrt{\epsilon}})\sqrt{\frac{N_l}{n_l}}$ and $h-1\geq\ln(\frac{2}{\sqrt{\epsilon}})\sqrt{\frac{N_l}{n_l}}$, respectively.

Thus, no matter $h$ is odd or even,  $P_h \geq 1-\epsilon$ holds for $h\geq\ln(\frac{2}{\sqrt{\epsilon}})\sqrt{\frac{N_l}{n_l}}+1$.  This completes the proof of Theorem~\ref{oneSetTheorem}.
\end{proof}
\begin{theorem}\label{twoSetTheorem}
In Algorithm~\ref{algorithm}, suppose that  the marked vertices are in both the left and right sides. There exists a sequence  of parameters $\alpha_i, \beta_i$,  such that if $h  \geq \ln(\frac{2}{\sqrt{\epsilon}})\max(\sqrt{\frac{N_l}{n_l}},\sqrt{\frac{N_r}{n_r}}) +1$,  then the algorithm will output a marked vertex with probability at least $1-\epsilon^2$, where $N_l$ ($N_r$) is the  number of left (right) vertices and $n_l$ ($n_r$) is the  number of marked vertices in the left (right) side.
\end{theorem}

\begin{proof}
Similar to the proof of Theorem \ref{oneSetTheorem}, by Lemma \ref{lemma3} one can show that
  \begin{itemize}
\item When $h$ is odd, in order to ensure $P_h \geq 1-\epsilon^2$, it suffices to satisfy
$h\geq\ln(\frac{2}{\sqrt{\epsilon}})\sqrt{\frac{N_l}{n_l}}$ and $h\geq\ln(\frac{2}{\sqrt{\epsilon}})\sqrt{\frac{N_r}{n_r}}$.

\item When $h$ is even, in order to ensure $P_h \geq 1-\epsilon^2$, it suffices to satisfy
$h+1\geq\ln(\frac{2}{\sqrt{\epsilon}})\sqrt{\frac{N_l}{n_l}}$, $h-1\geq\ln(\frac{2}{\sqrt{\epsilon}})\sqrt{\frac{N_r}{n_r}}$, $h-1\geq\ln(\frac{2}{\sqrt{\epsilon}})\sqrt{\frac{N_l}{n_l}}$ and $h+1\geq\ln(\frac{2}{\sqrt{\epsilon}})\sqrt{\frac{N_r}{n_r}}$.
 \end{itemize}
Therefore, no matter $h$ is even or odd,  $P_h \geq 1-\epsilon^2$ holds for $h\geq\ln(\frac{2}{\sqrt{\epsilon}})\max(\sqrt{\frac{N_l}{n_l}},\sqrt{\frac{N_r}{n_r}}) +1$.
This completes the proof of Theorem~\ref{twoSetTheorem}.
\end{proof}

\begin{lemma}\label{lemma2}
In Algorithm~\ref{algorithm}, suppose that all the marked vertices are in the left side.  There exists a  sequence  of parameters $\alpha_i, \beta_i$, such that the success probability satisfies
\begin{align*}
   P_h= 1-\epsilon T_h^2(\cos(\frac{1}{h}\arccos(\frac{1}{\sqrt{\epsilon}}))\sqrt{1-\frac{n_l}{N_l}})
\end{align*}
for odd $h$, and
\begin{align*}
  P_h=  &1-\frac{\epsilon}{2}(T_{h+1}^2(\cos(\frac{1}{h+1}\arccos(\frac{1}{\sqrt{\epsilon}}))\sqrt{1-\frac{n_l}{N_l}})\\ & +  T_{h-1}^2(\cos(\frac{1}{h-1}\arccos(\frac{1}{\sqrt{\epsilon}}))\sqrt{1-\frac{n_l}{N_l}})),
\end{align*}
for even $h$.
\end{lemma}

\begin{lemma}\label{lemma3}
 In Algorithm~\ref{algorithm}, suppose that  the marked vertices are in both the left and right sides. There exists a  sequence  of parameters $\alpha_i, \beta_i$, such that the success probability satisfies
\begin{align*}P_h= 1-&\epsilon^2 T_h^2(\cos(\frac{1}{h}\arccos(\frac{1}{\sqrt{\epsilon}}))\sqrt{1-\frac{n_l}{N_l}}))\\
&\times T_h^2(\cos(\frac{1}{h}\arccos(\frac{1}{\sqrt{\epsilon}}))\sqrt{1-\frac{n_r}{N_r}}))\label{phTwoSet-odd}
\end{align*}
for odd $h$, and
\begin{align*}
P_h &=1-\frac{\epsilon^2}{2}\bigg[ T_{h+1}^2\left(\cos\left(\frac{1}{h+1}\arccos(\frac{1}{\sqrt{\epsilon}})\right)\sqrt{1-\frac{n_l}{N_l}}\right)\\
  & \times T_{h-1}^2\Bigg(\cos\left(\frac{1}{h-1} \arccos(\frac{1}{\sqrt{\epsilon}})\right)\sqrt{1-\frac{n_r}{N_r}}\Bigg)\\
  & +T_{h+1}^2\left(\cos(\frac{1}{h+1}\arccos(\frac{1}{\sqrt{\epsilon}}))\sqrt{1-\frac{n_r}{N_r}}\right)\\
   & \times T_{h-1}^2\Bigg(\cos(\frac{1}{h-1}\arccos(\frac{1}{\sqrt{\epsilon}}))\sqrt{1-\frac{n_l}{N_l}}\Bigg)\bigg]
\end{align*}
for even $h$.
\end{lemma}

\section{Method}\label{method}
The section is devoted to the proof of Lemma \ref{lemma2}. As shown in Fig.~\ref{fig:1b}, vertices can be classified into three types: the marked vertices denoted by $u$ in the left, the unmarked vertices denoted by $v$ in the left and $s$ in the right. Therefore, our analysis can be simplified in a four-dimensional subspace with the  orthogonal basis $\{\ket{us}, \ket{su}, \ket{sv}, \ket{vs} \}$ given below:
\begin{align*}
&\ket{us}=\frac{1}{\sqrt{n_l}}\sum_u\ket{u}\otimes\frac{1}{\sqrt{N_r}}\sum_s\ket{s}, \\ &\ket{sv}=\frac{1}{\sqrt{N_r}}\sum_s\ket{s}\otimes\frac{1}{\sqrt{N_L-n_l}}\sum_v\ket{v}, \\
&\ket{su}=\frac{1}{\sqrt{N_r}}\sum_s\ket{s}\otimes\frac{1}{\sqrt{n_l}}\sum_u\ket{u},  \\ &\ket{vs}=\frac{1}{\sqrt{N_l-n_l}}\sum_v\ket{v}\otimes\frac{1}{\sqrt{N_r}}\sum_s\ket{s}.
\end{align*}

Note that $\ket{\Psi_0}$ can be rewritten in the above basis as
$\ket{\Psi_0}=\frac{1}{\sqrt{2N_lN_r}}\Big[\sqrt{n_lN_r}\ket{us}+\sqrt{n_lN_r}\ket{su}+ \sqrt{N_r(N_l-n_l)}\ket{sv}+
    \sqrt{N_r(N_l-n_l)}\ket{vs}\Big].$

Hence,  it can be expressed as a $4$-dimensional vector
$$\ket{\Psi_0}= \frac{1}{\sqrt{2N_lN_r}}\left(
    \begin{array}{c}
        \sqrt{n_lN_r} \\
        \sqrt{n_lN_r} \\
        \sqrt{N_r(N_l-n_l)} \\
       \sqrt{N_r(N_l-n_l)}\\
    \end{array}
    \right).$$
Furthermore, we have
\begin{align*}
    S=\left(
    \begin{array}{cccc}
        0 & 1 & 0 & 0 \\
        1 & 0 & 0 & 0 \\
        0 & 0 & 0 & 1 \\
        0 & 0 & 1 & 0 \\
    \end{array}
    \right),~~
    Q(\beta)=\left(
    \begin{array}{cccc}
        e^{i\beta} & 0 & 0 & 0 \\
        0 & 1 & 0 & 0 \\
        0 & 0 & 1 & 0 \\
        0 & 0 & 0 & 1 \\
    \end{array}
    \right),
\end{align*}
and
\begin{equation*}
    C(\alpha)=\left(
    \begin{array}{cccc}
        -e^{-i\alpha} & 0 & 0 & 0 \\
        0 & C_{22}& \frac{(1-e^{-i\alpha})\sin(\omega)}{2} & 0 \\
        0 & \frac{(1-e^{-i\alpha})\sin(\omega)}{2} & C_{33} & 0 \\
        0 & 0 & 0 & -e^{-i\alpha} \\
    \end{array}
    \right)
\end{equation*}
with $\omega= \arccos (1-\frac{2n_l}{N_l})$,
where  $C_{23}=\frac{(1-e^{-i\alpha})(1-\cos(\omega))}{2}-1$, $ C_{33} =\frac{(1-e^{-i\alpha})(1+\cos(\omega))}{2}-1,$
and
$\cos(\omega)=1-\frac{2n_l}{N_l}$, $\sin(\omega)=\frac{2}{N_L}\sqrt{n_l*(N_l-n_l)}$.

Now some key results are given as follows:
Let
\begin{equation*}
    R(\theta)=-\left(
    \begin{array}{cccc}
        e^{-\frac{i\theta}{2}} & 0 & 0 & 0 \\
        0 & e^{\frac{i\theta}{2}} & 0 & 0 \\
        0 & 0 & e^{-\frac{i\theta}{2}} & 0 \\
        0 & 0 & 0 & e^{-\frac{i\theta}{2}} \\
    \end{array}
    \right),
\end{equation*}
and
\begin{equation*}
    A(\theta)=\left(
    \begin{array}{cccc}
        1 & 0 & 0 & 0 \\
        0 & \cos{(\frac{\omega}{2})} & -ie^{i\theta}\sin{(\frac{\omega}{2})} & 0 \\
        0 & -ie^{-i\theta}\sin{(\frac{\omega}{2})} &  \cos{(\frac{\omega}{2})} & 0 \\
        0 & 0 & 0 & 1 \\
    \end{array}
    \right)
\end{equation*}
with $\omega= \arccos (1-\frac{2n_l}{N_l})$.
One can verify the following identities:
\begin{align}
  &C(\alpha)=e^{-\frac{i\alpha}{2}}A(\frac{\pi}{2})R(\alpha)A(-\frac{\pi}{2}) \label{C=ARA},
  \\&Q(\beta)S=-e^{i\frac{\beta}{2}}SR(\beta) \label{QS=SR},
  \\ &A(\alpha+\beta)=R(\beta)A(\alpha)R(-\beta)\label{A=RAR}, \\&R(\theta)R(-\theta)=I,\label{RR=I}\\
 &\ket{\Psi_0} =A(\frac{\pi}{2})SA(\frac{\pi}{2})\ket{\bar{0}} \label{ASA0},
\end{align}
where $\ket{\bar{0}}$ denotes $ (0,0,\frac{1}{\sqrt{2}},\frac{1}{\sqrt{2}})^T$.
Another crucial observation is the following lemma, which will be useful later:
\begin{lemma} \label{lemmasabs}
\begin{equation}\label{SABSCDS=CDSAB}
     SB_1SB_2S=B_2SB_1,
\end{equation}
where $B_1=\prod_{i=0}^n D_i$ and $B_2=\prod_{i=0}^m D_i$ for $D_i \in \{A(\theta_i),R(\theta_i)\}$.
\end{lemma}
\begin{proof}
By calculation, the form of $B_1$ and $B_2$  is as follows:
\begin{equation*}
    B_1=\left(
    \begin{array}{cccc}
        1 & 0 & 0 & 0 \\
        0 & B_1(22) & B_1(23) & 0 \\
        0 & B_1(32) &  B_1(33) & 0 \\
        0 & 0 & 0 & 1 \\
    \end{array}
    \right),
\end{equation*}
and
\begin{equation*}
    B_2=\left(
    \begin{array}{cccc}
        1 & 0 & 0 & 0 \\
        0 & B_2(22) & B_2(23) & 0 \\
        0 & B_2(32) &  B_2(33) & 0 \\
        0 & 0 & 0 & 1 \\
    \end{array}
    \right),
\end{equation*}
where $B_1(\cdot)$ and $B_2(\cdot)$ are mathematical expressions of $\theta$ and $\omega$.
Hence, we have $SB_1SB_2S=B_2SB_1$.
\end{proof}


Below we will prove Lemma \ref{lemma2} by two cases.

\noindent\textbf{Case 1: $h$ is an odd integer.}
First we set
\begin{equation}\label{oneSideOddA1}
    \alpha_k=
    \begin{cases}
    -\beta_{h+2-k}& \text{ $k=2,4,\dots,h-1$},\\
-\beta_{h-k}& \text{ $k=3,5,\dots,h$}.
    \end{cases}
\end{equation}
Then $\ket{\Psi_h}$ reduces  to
\begin{align}\label{oneSideOddreduced}
    \ket{\Psi_h}\sim
     S\big[A(\eta_h)...A(\eta_1)\big]R(\alpha_1)
   SR(\beta_h)\big[A(\zeta_h)...A(\zeta_1)\big]\ket{\bar{0}},
\end{align}
which will be proven in Appendix \ref{AppendixC} by using Eqs.~\eqref{C=ARA}-\eqref{SABSCDS=CDSAB}.
Here $\eta_{k}=\eta_{h+1-k}$ and $\zeta_{k}=\zeta_{h+1-k}$ for $k=1,2,\dots,h$,
and
\begin{equation}\label{oneSetOddEtaAlpha}
    \eta_{k+1}-\eta_{k}=
    \begin{cases}
    \pi - \alpha_{k+1} & \text{ $k=2,4,\dots,h-1$},\\
    -\pi+\alpha_{h-k+1} & \text{ $k=3,5,\dots,h$},
    \end{cases}
\end{equation}
\begin{equation}\label{oneSetOddZetaAlpha}
    \zeta_{k+1}-\zeta_{k}=
    \begin{cases}
    \pi - \alpha_{k} & \text{ $k=2,4,\dots,h-1$},\\
    -\pi+\alpha_{h-k} & \text{ $k=3,5,\dots,h$}.
    \end{cases}
\end{equation}

Let us have a more detailed analysis at the state evolution in Eq.~\eqref{oneSideOddreduced}, which can be divided into four stages as follows.

\begin{small}
\begin{align*}
\frac{1}{\sqrt{2}}
    \left(
    \begin{array}{c}
        0 \\
        0 \\
        1 \\
        1 \\
    \end{array}
    \right)
    &\xrightarrow[\textcircled 1]{A(\zeta_h)...A(\zeta_1)}
 \frac{1}{\sqrt{2}}
  \left(
    \begin{array}{c}
        0 \\
      b_h(x)\\
        c_h(x)\\
        1 \\
    \end{array}
    \right)
    \\
 &\xrightarrow [\textcircled 2]{\quad R(\alpha_1)SR(\beta_h)} \frac{1}{\sqrt{2}}
  \left(
    \begin{array}{c}
        e^{i\beta_h} b_h(x)\\
        0 \\
        1\\
         c_h(x)\\
    \end{array}
    \right)
\\
&\xrightarrow [\textcircled 3] {A(\eta_h)...A(\eta_1)}
\frac{1}{\sqrt{2}}
    \left(
    \begin{array}{c}
        e^{i\beta_h}b_h(x)\\
        \bar{b}_h(x) \\
        \bar{c}_h(x) \\
        c_h(x) \\
    \end{array}
    \right)
    \\
&\xrightarrow [\quad \textcircled 4] {\quad \quad S \quad \quad}
\frac{1}{\sqrt{2}}
 \left(
    \begin{array}{c}
        \bar{b}_h(x)\\
        e^{i\beta_h}b_h(x) \\
         c_h(x)\\
        \bar{c}_h(x) \\
    \end{array}
    \right).
\end{align*}
\end{small}

\noindent {\bf Stage $\textcircled 1$:  apply $A(\zeta_h)...A(\zeta_1)$  to the initial state.} Let
$(a_0, b_0, c_0, d_0)=(0,0,1,1)$ and
\begin{equation*}
   \ket{\mu_k}=
   \left(
    \begin{array}{c}
        a_k \\
        b_k \\
        c_k \\
        d_k \\
    \end{array}
    \right)=A(\zeta_k)...A(\zeta_1)
       \left(
    \begin{array}{c}
        0 \\
        0 \\
        1 \\
        1 \\
    \end{array}
    \right)
\end{equation*}
for $k=1,2,\dots,h.$ First note that $A(\zeta_i)$ makes effect only to the $2$th and $3$th dimension of a $4$-dimensional vector.  Thus,  $a_k=0$ and $d_k=1$ for all $k$. Furthermore,  we have
\begin{align}\label{uk}
    \ket{\mu_{k}}&=A(\zeta_k)\ket{\mu_{k-1}}\\ \nonumber
    &=\left(
        \begin{array}{c}
           0\\
            b_{k-1}\cos(\frac{\omega}{2})-ic_{k-1}e^{i\zeta_k}\sin(\frac{\omega}{2}) \\
            -ib_{k-1}e^{-i\zeta_k}\sin(\frac{\omega}{2})+c_{k-1}\cos(\frac{\omega}{2}) \\
           1 \\
        \end{array}
        \right),
\end{align}
and
\begin{align}\label{uk1}
    \ket{\mu_{k-2}}&=A(\zeta_{k-1})^{-1}\ket{\mu_{k-1}} \\ \nonumber
    &=\left(
        \begin{array}{c}
            0 \\
            b_{k-1}\cos(\frac{\omega}{2})+ic_{k-1}e^{i\zeta_{k-1}}\sin(\frac{\omega}{2}) \\
            ib_{k-1}e^{-i\zeta_{k-1}}\sin(\frac{\omega}{2})+c_{k-1}\cos(\frac{\omega}{2}) \\
           1\\
        \end{array}
        \right).
\end{align}
Combined with Eqs.~\eqref{uk} and~\eqref{uk1}, we have
\begin{equation*}
    c_k=-ib_{k-1}e^{-i\zeta_k}\sin(\frac{\omega}{2})+c_{k-1}\cos(\frac{\omega}{2}),
\end{equation*}
\begin{equation*}
    c_{k-2}=ib_{k-1}e^{-i\zeta_{k-1}}\sin(\frac{\omega}{2})+c_{k-1}\cos(\frac{\omega}{2}).
\end{equation*}
The recurrence formula of $c_k(x)$ is defined by
$c_0(x)=1,c_1(x)=x$, and for $k=2,...,h$,
\begin{equation*}
    c_k(x)=x(1+e^{-i(\zeta_k-\zeta_{k-1})})c_{k-1}(x)-e^{-i(\zeta_k-\zeta_{k-1})}c_{k-2}(x),
\end{equation*}
with $x=\cos(\frac{\omega}{2})$. By Lemma \ref{lemma1}, when
\begin{equation}\label{recurrence1a}
    \zeta_{k+1}-\zeta_{k}=(-1)^k\pi-2\arccot(\tan(\frac{k\pi}{h})\sqrt{1-\gamma^{2}})
\end{equation}
for $k=1,\dots,h-1$, where $\gamma^{-1}=\cos(\frac{1}{h}\arccos(\frac{1}{\sqrt{\epsilon}}))$, we have
\begin{equation*}
    c_h(x)=\frac{T_h(\frac{x}{\gamma})}{T_h(\frac{1}{\gamma})}
\end{equation*}
with $T_h(\frac{1}{\gamma})=\frac{1}{\sqrt{\epsilon}}$. Moreover, $b_h(x)$ is determined by $|b_h(x)|^2+|c_h(x)|^2=1$. Therefore, the state after $A(\zeta_h)...A(\zeta_1)$ applied to the initial state is
\begin{equation*}
\frac{1}{\sqrt{2}}
  \left(
    \begin{array}{c}
        0 \\
      b_h(x)\\
        c_h(x)\\
        1 \\
    \end{array}
    \right).
\end{equation*}

\noindent {\bf Stage $\textcircled 2$: apply $R(\alpha_1)SR(\beta_h)$ to the above state.} After that,
 the state is
\begin{equation*}
\frac{1}{\sqrt{2}}
  \left(
    \begin{array}{c}
        e^{i\beta_h} b_h(x)\\
        0 \\
        1\\
         c_h(x)\\
    \end{array}
    \right).
\end{equation*}

\noindent {\bf Stage $\textcircled 3$: perform $A(\eta_h)...A(\eta_1)$.}
Let
\begin{equation*}
    \left(
    \begin{array}{c}
        \bar{a}_k\\
        \bar{b}_k \\
        \bar{c}_k\\
        \bar{d}_k\\
    \end{array}
    \right)
    =A(\eta_k)...A(\eta_1)
  \left(
    \begin{array}{c}
        e^{i\beta_h} b_h(x)\\
        0 \\
        1\\
         c_h(x)\\
    \end{array}
    \right).
\end{equation*}
for $k=1,2,\dots,h.$
By the property of the matrix $A(\eta_{i})$, we have $\bar{a}_k=e^{i\beta_h} b_h(x)$
and $\bar{d}_k=c_h(x)$ for all $k$.
The recurrence formula of $\bar{c}_k(x)$ is defined by $\bar{c}_0(x)=1,\bar{c}_1(x)=x$ and for $k=2,...,h$,
\begin{equation*}
    \bar{c}_k(x)=x(1+e^{-i(\eta_k-\eta_{k-1})})\bar{c}_{k-1}(x)-e^{-i(\eta_k-\eta_{k-1})}\bar{c}_{k-2}(x),
\end{equation*}
with $x=\cos(\frac{\omega}{2})$.
By Lemma \ref{lemma1}, when
\begin{equation}\label{recurrence2}
    \eta_{k+1}-\eta_{k}=(-1)^k\pi-2cot^{-1}(\tan(\frac{k\pi}{h})\sqrt{1-\gamma^{2}})
\end{equation}
for $k=1,\dots,h-1$, where $\gamma^{-1}=\cos(\frac{1}{h}\arccos(\frac{1}{\sqrt{\epsilon}}))$, we have
\begin{equation*}
    \bar{c}_h(x)=\frac{T_h(\frac{x}{\gamma})}{T_h(\frac{1}{\gamma})}
\end{equation*}
with $T_h(\frac{1}{\gamma})=\frac{1}{\sqrt{\epsilon}}$.
Moreover, $\bar{b}_h(x)$ is determined by $|\bar{b}_h(x)|^2+|\bar{c}_h(x)|^2=1$.
Hence, the result state is
\begin{equation*}
\frac{1}{\sqrt{2}}
    \left(
    \begin{array}{c}
        e^{i\beta_h}b_h(x)\\
        \bar{b}_h(x) \\
        \bar{c}_h(x) \\
        c_h(x) \\
    \end{array}
    \right).
\end{equation*}

\noindent {\bf Stage $\textcircled 4$: perform  the final operation $S$.}
The  final state is
\begin{equation*}
\frac{1}{\sqrt{2}}
    \left(
    \begin{array}{c}
        \bar{b}_h(x)\\
        e^{i\beta_h}b_h(x) \\
         c_h(x)\\
        \bar{c}_h(x) \\
    \end{array}
    \right).
\end{equation*}

Therefore, the success probability $P_{h}$ is
\begin{align*}
        P_{h}&=1-\frac{1}{2}(|c_h(x)|^2+|\bar{c_h}(x)|^2) = 1- \epsilon T_h^2(\frac{x}{\gamma}) \\
    &= 1-\epsilon T_h^2(\cos(\frac{1}{h}\arccos(\frac{1}{\sqrt{\epsilon}}))\sqrt{1-\frac{n_l}{N_l}}).
\end{align*}
By Eqs.~\eqref{oneSideOddA1},~\eqref{oneSetOddEtaAlpha},
~\eqref{oneSetOddZetaAlpha},~\eqref{recurrence1a}  and~\eqref{recurrence2}, $\alpha_k, \beta_k$ can be chosen such that
\begin{align*}
    \alpha_k&=-\beta_{h+2-k}=\pi+(\eta_{k+1}-\eta_{k})\\
    &=2\arccot(\tan(\frac{k\pi}{h})\sqrt{1-\gamma^2})
\end{align*}
for $k=2,4,\dots,h-1$, and
\begin{align*}
    \alpha_k&=-\beta_{h-k}=\pi-(\zeta_{k+1}-\zeta_{k})\\
    &=2\arccot(\tan(\frac{(k-1)\pi}{h})\sqrt{1-\gamma^2})
\end{align*}
for $k=3,5,\dots,h$. In addition, $\alpha_1$ and $\beta_h$ can be any value.

\noindent\textbf{Case 2: $h$ is an even integer.} The proof is given in Appendix \ref{AppendixA}.

\section{Conclusion \& outlook}
In this paper,  we  investigated how to overcome the souffl\'{e} problem of  quantum walk search. We presented a robust quantum walk-based algorithm for searching a marked vertex on a complete bipartite graph. The algorithm  need not  know any prior information about the marked vertices (e.g., the number of marked vertices), but  keeps   a   quadratic speedup over classical search algorithms and ensures that the error is bounded by a tunable parameter $\epsilon$.

We have just initiated the first step towards robust quantum walk search. More questions are worthy of further consideration.  For example, several interesting questions are listed below:
\begin{itemize}
 \item Can the robustness feature be introduced into quantum walk search on other graphs?
  \item For the framework of searching a marked state in  Markov chains, can we propose a robust version?
  \item Another interesting direction would be to explore some important
problems in  practical scenarios by
robust quantum walk search algorithms.
\end{itemize}

We will try to address these questions in forthcoming works.

\section*{Acknowledgments}
This work was supported by the National Natural Science Foundation of China (Grant Nos. 62272492, 61772565), the Guangdong Basic and Applied Basic Research Foundation (Grant No. 2020B1515020050).

\bibliography{refs}

\begin{thebibliography}{42}
\bibitem [{sof()}]{soffle}%
  \BibitemOpen
  \href@noop {} {}\bibinfo {note} {Grover's quantum searching technique is like
  cooking a soufflé. You put the state obtained by quantum parallelism in a
  ``quantum oven'' and let the desired answer rise slowly. Success is almost
  guaranteed if you open the oven at just the right time. But the soufflé is
  very likely to fall---the amplitude of the correct answer drops to zero---if
  you open the oven too early. Furthermore, the soufflé could burn if you
  overcook it: strangely, the amplitude of the desired state starts shrinking
  after reaching its maximum. After twice the optimal number of shakes, you are
  no more likely to succeed than before the first shake.}\BibitemShut {Stop}%
\end{thebibliography}
\appendix
\begin{widetext}
\newpage
\section{Case 2 in the proof of Lemma \ref{lemma2}: $h$ is an even integer} \label{AppendixA}

First we set
\begin{equation}\label{oneSideEvenA1}
    \beta_k=
    -\alpha_{h+1-k} \text{ \quad \quad $k=1,2,\dots,h-1$}.
\end{equation}
Then $\ket{\Psi_h}$ reduces to

\begin{align}\label{oneSideEvenreduced}
    \ket{\Psi_h} \sim &R(\beta_h)\big[A(\phi_{h-1})...A(\phi_1)\big]
 R(\alpha_1)S\big[A(\psi_{h+1})...A(\psi_1)\big]\ket{\bar{0}},
\end{align}
which will be proven in Appendix \ref{AppendixC} by using Eqs.~\eqref{C=ARA}-\eqref{SABSCDS=CDSAB}.
Here $\phi_{k}=\phi_{h-k}$ for $k=1,2,\dots,h-1$, $\psi_{k}=\psi_{h+2-k}$
for $k=1,2,\dots,h+1$,
and
\begin{equation}\label{oneSideEvenA2}
    \phi_{k+1}-\phi_{k}=
    \begin{cases}
    \pi - \alpha_{k+1} & \text{ $k=2,4,\dots,h-2$},\\
    -\pi+\alpha_{h-k} & \text{ $k=1,3,\dots,h-3$},
    \end{cases}
\end{equation}
\begin{equation}\label{oneSideEvenA3}
    \psi_{k+1}-\psi_{k}=
    \begin{cases}
    \pi - \alpha_{k} & \text{$k=2,4,\dots,h-2$},\\
    -\pi+\alpha_{h-k+1} & \text{ $k=1,3,\dots,h$}.
    \end{cases}
\end{equation}
Similar to Eq.~\eqref{oneSideOddreduced}, the final state of Eq.~\eqref{oneSideEvenreduced} can be obtained by the following four stages:
\begin{small}
 \begin{align*}
     \frac{1}{\sqrt{2}}
    \left(
    \begin{array}{c}
        0 \\
        0 \\
        1 \\
        1 \\
    \end{array}
    \right)
      & \xrightarrow [\textcircled 1]{A(\psi_{h+1})...A(\psi_1)}
 \frac{1}{\sqrt{2}}
    \left(
    \begin{array}{c}
        0\\
        b_{h+1}(x) \\
        c_{h+1}(x) \\
        1 \\
    \end{array}
    \right)
\xrightarrow [\textcircled 2]{R(\alpha_1)S}
\frac{1}{\sqrt{2}}
    \left(
    \begin{array}{c}
        b_{h+1}(x)\\
         0\\
        1 \\
        c_{h+1}(x) \\
    \end{array}
    \right)
\\
&\xrightarrow [\textcircled 3]{A(\phi_{h-1})...A(\phi_1)}
\frac{1}{\sqrt{2}}
    \left(
    \begin{array}{c}
        b_{h+1}(x)\\
         \bar{b}_{h-1}(x)\\
        \bar{c}_{h-1}(x) \\
        c_{h+1}(x) \\
    \end{array}
    \right)
\xrightarrow [\textcircled 4]{R(\beta_h)}
\frac{1}{\sqrt{2}}
    \left(
    \begin{array}{c}
        b_{h+1}(x)\\
        e^{i\beta_h} \bar{b}_{h-1}(x)\\
        \bar{c}_{h-1}(x) \\
        c_{h+1}(x) \\
    \end{array}
    \right)
.
 \end{align*}
 \end{small}

\noindent {\bf Stage $\textcircled 1$:  apply $A(\psi_{h+1})...A(\psi_1)$  to the initial state.} Let
$(a_0, b_0, c_0, d_0)=(0,0,1,1)$ and
\begin{equation*}
 \left(
    \begin{array}{c}
        a_k \\
        b_k \\
        c_k \\
        d_k \\
    \end{array}
    \right)=A(\psi_k)...A(\psi_1)
     \left(
    \begin{array}{c}
        0 \\
        0 \\
        1 \\
        1 \\
    \end{array}
    \right)
\end{equation*}
for $k=1,2,\dots,h+1.$ By the property of  matrix $A(\psi_i)$, we have $a_k=0$ and $d_k=1$ for all $k$.
 The recurrence formula of $c_k(x)$ is defined by $c_0(x)=1,c_1(x)=x$ and for $k=2,...,h+1$,
\begin{equation*}
    c_k(x)=x(1+e^{-i(\psi_k-\psi_{k-1})})c_{k-1}(x)-e^{-i(\psi_k-\psi_{k-1})}c_{k-2}(x),
\end{equation*} with $x=\cos(\frac{\omega}{2})$.
By Lemma \ref{lemma1}, when
\begin{equation}\label{oneSideEvenA4}
    \psi_{k+1}-\psi_{k}=(-1)^k\pi-2\arccot(\tan(\frac{k\pi}{h+1})\sqrt{1-\gamma_1^{2}})
 \end{equation}
 with $\gamma^{-1}_1=\cos(\frac{1}{h+1}\arccos(\frac{1}{\sqrt{\epsilon}}))$ for $k=1,2,\dots,h$,
we have
 \begin{equation*}
     c_{h+1}(x)=\frac{T_{h+1}(\frac{x}{\gamma_1})}{T_{h+1}(\frac{1}{\gamma_1})},
 \end{equation*}
with $T_{h+1}(\frac{1}{\gamma_1})=\frac{1}{\sqrt{\epsilon}}$.
Moreover, $b_{h+1}(x)$ is determined by $|b_{h+1}(x)|^2+|c_{h+1}(x)|^2=1$. Therefore, the state after $A(\psi_{h+1})...A(\psi_1)$ applied to the initial state is
\begin{equation*}
\frac{1}{\sqrt{2}}
    \left(
    \begin{array}{c}
        0\\
        b_{h+1}(x) \\
        c_{h+1}(x) \\
        1 \\
    \end{array}
    \right).
\end{equation*}

\noindent {\bf Stage $\textcircled 2$: apply $R(\alpha_1)S$ to the above state.} After that,
 the state is
\begin{equation*}
\frac{1}{\sqrt{2}}
    \left(
    \begin{array}{c}
        b_{h+1}(x)\\
         0\\
        1 \\
        c_{h+1}(x) \\
    \end{array}
    \right).
\end{equation*}

\noindent {\bf Stage $\textcircled 3$: perform $A(\phi_{h-1})...A(\phi_1)$.} Let
\begin{equation*}
  \left(
    \begin{array}{c}
        \bar{a}_k\\
        \bar{b}_k \\
        \bar{c}_k\\
        \bar{d}_k\\
    \end{array}
    \right)
    =A(\phi_k)...A(\phi_1)
    \left(
    \begin{array}{c}
        b_{h+1}(x)\\
         0\\
        1 \\
        c_{h+1}(x) \\
    \end{array}
    \right)
\end{equation*}
for $k=1,2,\dots,h-1.$
By the property of  matrix $A(\phi_{i})$, we have $\bar{a}_k=b_{h+1}(x)$
and $\bar{d}_k=c_{h+1}(x)$ for all $k$.
The recurrence formula of $\bar{c}_k(x)$ is defined by $\bar{c}_0(x)=1,\bar{c}_1(x)=x$ and for $k=2,...,h-1$,
\begin{equation*}
    \bar{c}_k(x)=x(1+e^{-i(\phi_k-\phi_{k-1})})\bar{c}_{k-1}(x)-e^{-i(\phi_k-\phi_{k-1})}\bar{c}_{k-2}(x)
\end{equation*}
with $x=\cos(\frac{\omega}{2})$. By Lemma \ref{lemma1}, when
\begin{equation}\label{oneSideEvenA5}
    \phi_{k+1}-\phi_{k}=(-1)^k\pi-2\arccot(\tan(\frac{k\pi}{h-1})\sqrt{1-\gamma_2^{2}})
 \end{equation}
 with $\gamma^{-1}_2=\cos(\frac{1}{h-1}\arccos(\frac{1}{\sqrt{\epsilon}}))$
 for $k=1,2,\dots,h-2,$
 we have
\begin{equation*}
 \bar{c}_{h-1}(x)=\frac{T_{h-1}(\frac{x}{\gamma_2})}{T_{h-1}(\frac{1}{\gamma_2})},
 \end{equation*}
with $T_{h-1}(\frac{1}{\gamma_2})=\frac{1}{\sqrt{\epsilon}}$.
Moreover, $\bar{b}_{h-1}(x)$ is determined by $|\bar{b}_{h-1}(x)|^2+|\bar{c}_{h-1}(x)|^2=1$. Hence, the result state is
\begin{equation*}
\frac{1}{\sqrt{2}}
    \left(
    \begin{array}{c}
        b_{h+1}(x)\\
         \bar{b}_{h-1}(x)\\
        \bar{c}_{h-1}(x) \\
        c_{h+1}(x) \\
    \end{array}
    \right).
\end{equation*}

\noindent {\bf Stage $\textcircled 4$: perform  the final operation $R(\beta_h)$.}
The  final state is
\begin{equation*}
\frac{1}{\sqrt{2}}
    \left(
    \begin{array}{c}
        b_{h+1}(x)\\
        e^{i\beta_h} \bar{b}_{h-1}(x)\\
        \bar{c}_{h-1}(x) \\
        c_{h+1}(x) \\
    \end{array}
    \right).
\end{equation*}

Therefore, the success probability $P_{h}$ is
\begin{align*}
  P_{h} &= 1- \frac{1}{2}(|\bar{c}_{h-1}(x)|^2+|c_{h-1}(x)|^2)\\
   & =1-\frac{\epsilon}{2}(T^2_{h+1}(\frac{x}{\gamma_1})+T^2_{h-1}(\frac{x}{\gamma_2}))\\
 &   =1-\frac{\epsilon}{2}\bigg[T_{h+1}^2(\cos(\frac{1}{h+1}\arccos(\frac{1}{\sqrt{\epsilon}}))\sqrt{1-\frac{n_l}{N_l}})+T_{h-1}^2(\cos(\frac{1}{h-1}\arccos(\frac{1}{\sqrt{\epsilon}}))\sqrt{1-\frac{n_l}{N_l}})\bigg].
\end{align*}

By Eq.~\eqref{oneSideEvenA1} and Eqs.~\eqref{oneSideEvenA2}-\eqref{oneSideEvenA5}, $\alpha_k, \beta_k$ can be chosen such that
\begin{equation*}
    \alpha_k=-\beta_{h+1-k}=\pi -(\phi_k-\phi_{k-1})\\=2\arccot(\tan(\frac{k\pi}{h+1})\sqrt{1-\gamma_1^2})
\end{equation*}
for $k=2,4,\dots,h$, and
\begin{equation*}
    \alpha_k=-\beta_{h+1-k}=\pi -(\psi_{k+1}-\psi_k)\\=2\arccot(\tan(\frac{(k-1)\pi}{h-1})\sqrt{1-\gamma_2^2})
\end{equation*}
for $k=3,5,\dots,h-1$.
In addition, $\alpha_1$ and $\beta_h$ can be any value.

\section{Proof of Lemma \ref{lemma3}: marked vertices in two sides}\label{AppendixB}
The section is devoted to the proof of Lemma \ref{lemma3}.

\begin{proof}
Vertices in Fig.~\ref{fig:1c} can be divided into four types: the marked vertices denoted by $u$ in the left and $t$ in the right, the unmarked vertices denoted by $v$ in the left and $s$ in the right.
Hence, our analysis can be simplified in  an 8-dimensional subspace defined by the following orthogonal basis $\{\ket{ut},\ket{us},\ket{tu},\ket{tv},\ket{vt},\ket{vs},\ket{su},\ket{sv}\}$:
\begin{align*}
    &\ket{ut}=\frac{1}{\sqrt{n_l}}\sum_u\ket{u}\otimes\frac{1}{\sqrt{n_r}}\sum_t\ket{t},
    \quad \ket{us}=\frac{1}{\sqrt{n_l}}\sum_u\ket{u}\otimes\frac{1}{\sqrt{N_r-n_r}}\sum_s\ket{s},\\
   & \ket{tu}=\frac{1}{\sqrt{n_r}}\sum_t\ket{t}\otimes\frac{1}{\sqrt{n_l}}\sum_u\ket{u},
     \quad \ket{tv}=\frac{1}{\sqrt{n_r}}\sum_t\ket{t}\otimes\frac{1}{\sqrt{N_l-n_l}}\sum_v\ket{v},\\
    &  \ket{vt}=\frac{1}{\sqrt{N_l-n_l}}\sum_v\ket{v}\otimes\frac{1}{\sqrt{n_r}}\sum_t\ket{t},\\
     &\ket{vs}=\frac{1}{\sqrt{N_l-n_l}}\sum_v\ket{v}\otimes\frac{1}{\sqrt{N_r-n_r}}\sum_s\ket{s},\\
   & \ket{su}=\frac{1}{\sqrt{N_r-n_r}}\sum_s\ket{s}\otimes\frac{1}{\sqrt{n_l}}\sum_u\ket{u},\\
   & \ket{sv}=\frac{1}{\sqrt{N_r-n_r}}\sum_s\ket{s}\otimes\frac{1}{\sqrt{N_l-n_l}}\sum_v\ket{v}.
\end{align*}

Then $\ket{\Psi_0}$ can be rewritten in the above basis as
\begin{align*}
\ket{\Psi_0}
&=\frac{1}{\sqrt{2N_lN_r}}\Big[\sqrt{n_ln_r}\ket{ut}+\sqrt{n_l(N_r-n_r)}\ket{us}+
\sqrt{n_ln_r}\ket{tu}\\
&+\sqrt{n_r(N_l-n_l)}\ket{tv}+\sqrt{n_r(N_l-n_l)}\ket{vt}+\sqrt{(N_l-n_l)(N_r-n_r)}\ket{vs}\\
&+\sqrt{n_l(N_r-n_r)}\ket{su}+\sqrt{(N_l-n_l)(N_r-n_r)}\ket{sv} \Big].
\end{align*}
Thus, $\ket{\Psi_0}$  can be expressed as an $8$-dimensional vector
$$
\ket{\Psi_0}= \frac{1}{\sqrt{2N_lN_r}}
\begin{pmatrix}
    \begin{smallmatrix}
        \sqrt{n_ln_r} \\
        \sqrt{n_l(N_r-n_r)} \\
        \sqrt{n_ln_r} \\
        \sqrt{n_r(N_l-n_l)} \\
        \sqrt{n_r(N_l-n_l)} \\
        \sqrt{(N_l-n_l)(N_r-n_r)} \\
        \sqrt{n_l(N_r-n_r)}\\
        \sqrt{(N_l-n_l)(N_r-n_r)} \\
   \end{smallmatrix}
    \end{pmatrix}.
$$

Furthermore, we have
\begin{align*}
    S=
     \begin{pmatrix}
    \begin{smallmatrix}
        0 & 0 & 1 & 0 & 0 & 0 & 0 & 0\\
        0 & 0 & 0 & 0 & 0 & 0 & 1 & 0\\
        1 & 0 & 0 & 0 & 0 & 0 & 0 & 0\\
        0 & 0 & 0 & 0 & 1 & 0 & 0 & 0\\
        0 & 0 & 0 & 1 & 0 & 0 & 0 & 0\\
        0 & 0 & 0 & 0 & 0 & 0 & 0 & 1\\
        0 & 1 & 0 & 0 & 0 & 0 & 0 & 0\\
        0 & 0 & 0 & 0 & 0 & 1 & 0 & 0\\
     \end{smallmatrix}
    \end{pmatrix},
    Q(\beta)=\begin{pmatrix}
    \begin{smallmatrix}
        e^{i\beta} & 0 & 0 & 0 & 0 & 0 & 0 & 0\\
        0 & e^{i\beta} & 0 & 0 & 0 & 0 & 0 & 0\\
        0 & 0 & e^{i\beta} & 0 & 0 & 0 & 0 & 0\\
        0 & 0 & 0 & e^{i\beta} & 0 & 0 & 0 & 0\\
        0 & 0 & 0 & 0 & 1 & 0 & 0 & 0 \\
        0 & 0 & 0 & 0 & 0 & 1 & 0 & 0 \\
        0 & 0 & 0 & 0 & 0 & 0 & 1 & 0 \\
        0 & 0 & 0 & 0 & 0 & 0 & 0 & 1 \\
     \end{smallmatrix}
    \end{pmatrix},
\end{align*}
and
\begin{small}
\begin{equation*}
    C(\alpha)=\left(
    \begin{array}{cc}
       1 & 0 \\
       0 & 1
    \end{array}
    \right)
    \otimes
    \left(
    \begin{array}{cccc}
       M_1& M_2 & 0 & 0\\
        M_2 & M_3 & 0 & 0\\
        0 & 0 &  M_4 & M_5\\
        0 & 0 & M_5 & M_6
    \end{array}
    \right),
\end{equation*}
\end{small}
with
\begin{align*}
& M_1=\frac{(1-e^{-i\alpha})(1-\cos(\omega_2))}{2}-1 ,  ~~~~~ M_2=\frac{(1-e^{-i\alpha})\sin(\omega_2)}{2},\\
&  M_3=\frac{(1-e^{-i\alpha})(1+\cos(\omega_2))}{2}-1, ~~~~~
 M_4=\frac{(1-e^{-i\alpha})(1-\cos(\omega_1))}{2}-1,  \\ & M_5=\frac{(1-e^{-i\alpha})\sin(\omega_1)}{2},  ~~~~~~~~~~~~~~~~~~
   M_6=\frac{(1-e^{-i\alpha})(1+\cos(\omega_1))}{2}-1,
\end{align*}
where $\cos(\omega_1)=1-\frac{2n_l}{N_l}$,  $\sin(\omega_1)=\frac{2}{N_l}\sqrt{n_l(N_l-n_l)}$,  $\cos(\omega_2)=1-\frac{2n_r}{N_r}$, and $\sin(\omega_2)=\frac{2}{N_r}\sqrt{n_r(N_r-n_r)}$.

Let
\begin{equation*}
    R(\theta)=-\begin{pmatrix}
    \begin{smallmatrix}
        e^{\frac{i\theta}{2}} & 0 & 0 & 0 & 0 & 0 & 0 & 0\\
        0 & e^{-\frac{i\theta}{2}} & 0 & 0 & 0 & 0 & 0 & 0\\
        0 & 0 & e^{\frac{i\theta}{2}} & 0 & 0 & 0 & 0 & 0\\
        0 & 0 & 0 & e^{-\frac{i\theta}{2}} & 0 & 0 & 0 & 0\\
        0 & 0 & 0 & 0 & e^{\frac{i\theta}{2}} & 0 & 0 & 0 \\
        0 & 0 & 0 & 0 & 0 & e^{-\frac{i\theta}{2}} & 0 & 0 \\
        0 & 0 & 0 & 0 & 0 & 0 & e^{\frac{i\theta}{2}} & 0 \\
        0 & 0 & 0 & 0 & 0 & 0 & 0 & e^{-\frac{i\theta}{2}} \\
    \end{smallmatrix}
    \end{pmatrix},
\end{equation*}
and
\begin{equation*}
    A(\theta)=\left(
    \begin{array}{cc}
       1 & 0 \\
       0 & 1
    \end{array}
    \right)
    \otimes
    \begin{pmatrix}
    \begin{smallmatrix}
        \cos(\frac{\omega_2}{2}) & -ie^{i\theta}\frac{\sin(\omega_2)}{2} & 0 & 0\\
        -ie^{-i\theta}\frac{\sin(\omega_2)}{2} & \cos(\frac{\omega_2}{2}) & 0 & 0\\
        0 & 0 &  \cos(\frac{\omega_1}{2}) & -ie^{i\theta}\frac{\sin(\omega_1)}{2}\\
        0 & 0 & -ie^{-i\theta}\frac{\sin(\omega_1)}{2} & \cos(\frac{\omega_1}{2})
    \end{smallmatrix}
    \end{pmatrix}.
\end{equation*}
We have
\begin{align}
  &C(\alpha)=e^{-\frac{i\alpha}{2}}A(\frac{\pi}{2})R(\alpha)A(-\frac{\pi}{2}) \label{C11},
  \\&Q(\beta)S=-e^{i\frac{\beta}{2}}SR(\beta) \label{C12},
  \\ &A(\alpha+\beta)=R(\beta)A(\alpha)R(-\beta) \label{C13},
  \\&R(\theta)R(-\theta)=I \label{C14},\\
 &\ket{\Psi_0} =A(\frac{\pi}{2})SA(\frac{\pi}{2})\ket{\bar{0}} \label{C15},
\end{align}
where $\ket{\bar{0}}$ denotes $ (0,0,0,0,0,\frac{1}{\sqrt{2}},0,\frac{1}{\sqrt{2}})^T$.
Similarly to Eq.~\eqref{SABSCDS=CDSAB}, the following equation holds.
\begin{equation}\label{C16}
\begin{split}
    SB_1SB_2S=B_2SB_1,
\end{split}
\end{equation}
where $B_1=\prod_{i=0}^n D_i$ and $B_2=\prod_{i=0}^m D_i$ for $D_i \in \{A(\theta_i),R(\theta_i)\}$.

Below we will prove Lemma \ref{lemma3} by two cases.

\end{proof}

\subsection{Case 1: $h$ is an odd integer}\label{subsection5.1}

First we set
\begin{equation}\label{twoSideOddA1}
    \alpha_k=
    \begin{cases}
    -\beta_{h+2-k}& \text{ $k=2,4,\dots,h-1$},\\
-\beta_{h-k}& \text{ $k=3,5,\dots,h$}.
    \end{cases}
\end{equation}
Then $\ket{\Psi_h}$ reduces  to
\begin{align}\label{twoSideOddReduce}
    \ket{\Psi_h}\sim
     S\big[A(\eta_h)...A(\eta_1)\big]R(\alpha_1)
   SR(\beta_h)\big[A(\zeta_h)...A(\zeta_1)\big]\ket{\bar{0}},
\end{align}
which will be proven in Appendix \ref{AppendixD} by using Eqs.~\eqref{C11}-\eqref{C16}. Here $\eta_{k}=\eta_{h+1-k}$, $\zeta_{k}=\zeta_{h+1-k}$ for $k=1,2,\dots,h$,
and
\begin{equation}\label{twoSetOddEtaAlpha}
    \eta_{k+1}-\eta_{k}=
    \begin{cases}
    \pi - \alpha_{k+1} & \text{ $k=2,4,\dots,h-1$},\\
    -\pi+\alpha_{h-k+1} & \text{ $k=3,5,\dots,h$},
    \end{cases}
\end{equation}
\begin{equation}\label{twoSetOddZetaAlpha}
    \zeta_{k+1}-\zeta_{k}=
    \begin{cases}
    \pi - \alpha_{k} & \text{ $k=2,4,\dots,h-1$},\\
    -\pi+\alpha_{h-k} & \text{ $k=3,5,\dots,h$}.
    \end{cases}
\end{equation}
The final state of Eq.~\eqref{twoSideOddReduce} can be obtained by the following four stages:
\begin{align*}
     \frac{1}{\sqrt{2}} \begin{pmatrix}
    \begin{smallmatrix}
        0 \\
        0 \\
        0 \\
        0 \\
        0 \\
       1 \\
        0 \\
        1
   \end{smallmatrix}
    \end{pmatrix}
&\xrightarrow [\textcircled 1] {\text{\small $A(\zeta_h)...A(\zeta_1)$}}
 \frac{1}{\sqrt{2}}
 \begin{pmatrix}
    \begin{smallmatrix}
        0 \\
        0 \\
        0 \\
        0 \\
        e_h(x_2)\\
      f_h(x_2) \\
        g_h(x_1) \\
        l_h(x_1)
     \end{smallmatrix}
    \end{pmatrix}
\xrightarrow [\textcircled 2] {R(\alpha_1)SR(\beta_h)}
  \frac{1}{\sqrt{2}} \begin{pmatrix}
    \begin{smallmatrix}
        0 \\
        e^{i\beta_h}g_h(x_1) \\
        0 \\
        e^{i\beta_h}e_h(x_2) \\
        0 \\
        g_h(x_1) \\
        0 \\
        f_h(x_2)
   \end{smallmatrix}
    \end{pmatrix}
\\
&\xrightarrow [\textcircled 3]{A(\eta_h)...A(\eta_1)}
\frac{1}{\sqrt{2}} \begin{pmatrix}
    \begin{smallmatrix}
        e^{i\beta_h}g_h(x_1)\bar{a}_h(x_2) \\
        e^{i\beta_h}g_h(x_1)\bar{b}_h(x_2) \\
        e^{i\beta_h}e_h(x_2)\bar{c}_h(x_1) \\
        e^{i\beta_h}e_h(x_2)\bar{d}_h(x_1) \\
        g_h(x_1)\bar{e}_h(x_2) \\
        g_h(x_1)\bar{f}_h(x_2) \\
        f_h(x_2)\bar{g}_h(x_1) \\
        f_h(x_2)\bar{l}_h(x_1)
   \end{smallmatrix}
    \end{pmatrix}\\
 & \xrightarrow [\textcircled 4]{\quad\quad S \quad\quad}
\frac{1}{\sqrt{2}}
\begin{pmatrix}
    \begin{smallmatrix}
         e^{i\beta_h}e_h(x_2)\bar{c}_h(x_1)\\
        f_h(x_2)\bar{g}_h(x_1) \\
        e^{i\beta_h}g_h(x_1)\bar{a}_h(x_2) \\
         g_h(x_1)\bar{e}_h(x_2)\\
       e^{i\beta_h}e_h(x_2)\bar{d}_h(x_1)  \\
       f_h(x_2)\bar{l}_h(x_1) \\
        e^{i\beta_h}g_h(x_1)\bar{b}_h(x_2) \\
         g_h(x_1)\bar{f}_h(x_2)
   \end{smallmatrix}
    \end{pmatrix}.
\end{align*}

\noindent {\bf Stage $\textcircled 1$:  apply $A(\zeta_{h})...A(\zeta_1)$  to the initial state.} Let $$(a_0, b_0, c_0, d_0, e_0, f_0, g_0, l_0)=(0,0,0,0,0,1,0,1)$$ and
\begin{equation*}
   \begin{pmatrix}
    \begin{smallmatrix}
        a_k \\
        b_k \\
        c_k \\
        d_k \\
        e_k \\
      f_k \\
       g_k \\
        l_k
     \end{smallmatrix}
    \end{pmatrix}
    =A(\zeta_{k})...A(\zeta_1)
   \begin{pmatrix}
    \begin{smallmatrix}
        0 \\
        0 \\
        0 \\
        0 \\
        0 \\
      1\\
        0 \\
        1
     \end{smallmatrix}
    \end{pmatrix}
\end{equation*}
for $k=1,2,\dots,h.$
The elements of a $8$-dimensional vector are divided into four groups: (1th, 2th), (3th, 4th), (5th, 6th), and (7th, 8th). Note that each block of  matrix $A(\zeta_{i})$ acts on a corresponding group of a $8$-dimensional vector. Thus, $a_k=b_k=c_k=d_k=0$ for all $k.$
The recurrence formula of $l_k(x_1)$ is defined by $l_0(x_1)=1,l_1(x_1)=x_1$ and for $k=2,...,h$
\begin{equation*}
    l_k(x_1)=x_1(1+e^{-i(\zeta_k-\zeta_{k-1})})l_{k-1}(x_1)-e^{-i(\zeta_k-\zeta_{k-1})}l_{k-2}(x_1)
\end{equation*} with $x_1=\cos(\frac{\omega_1}{2})$.
The recurrence formula of $f_k(x_2)$ is defined by $f_0(x_2)=1,f_1(x_2)=x_2$ and for $k=2,...,h$
\begin{equation*}
    f_k(x_2)=x_2(1+e^{-i(\zeta_k-\zeta_{k-1})})f_{k-1}(x_2)-e^{-i(\zeta_k-\zeta_{k-1})}f_{k-2}(x_2)
\end{equation*} with $x_2=\cos(\frac{\omega_2}{2})$.
By Lemma \ref{lemma1}, when
\begin{equation}\label{twoSideOddA4}
  \zeta_{k+1}-\zeta_{k}=(-1)^k\pi-2\arccot(\tan(\frac{k\pi}{h})\sqrt{1-\gamma^{2}}),
\end{equation}
with $\gamma^{-1}=\cos(\frac{1}{h}\arccos(\frac{1}{\sqrt{\epsilon}}))$ for $k=1,2,\dots,h-1,$ we have
\begin{equation*}
 l_h(x_1)=\frac{T_h(\frac{x_1}{\gamma})}{T_h(\frac{1}{\gamma})},  f_h(x_2)=\frac{T_h(\frac{x_2}{\gamma})}{T_h(\frac{1}{\gamma})},
\end{equation*}
with $T_h(\frac{1}{\gamma})=\frac{1}{\sqrt{\epsilon}}$.
Moreover, $e_h(x_2)$ and $g_h(x_1)$ are determined by $|e_h(x_2)|^2+|f_h(x_2)|^2=1$ and $|g_h(x_1)|^2+|l_h(x_1)|^2=1$, respectively. Therefore, the state after $A(\zeta_{h})...A(\zeta_1)$ applied to the initial state is
\begin{equation*}
 \frac{1}{\sqrt{2}}
 \begin{pmatrix}
    \begin{smallmatrix}
        0 \\
        0 \\
        0 \\
        0 \\
        e_h(x_2)\\
      f_h(x_2) \\
        g_h(x_1) \\
        l_h(x_1)
     \end{smallmatrix}
    \end{pmatrix}.
\end{equation*}

\noindent {\bf Stage $\textcircled 2$: apply $R(\alpha_1)SR(\beta_h)$ to the above state.} After that, the state is
\begin{equation*}
    \frac{1}{\sqrt{2}} \begin{pmatrix}
    \begin{smallmatrix}
        0 \\
        e^{i\beta_h}g_h(x_1) \\
        0 \\
        e^{i\beta_h}e_h(x_2) \\
        0 \\
        g_h(x_1) \\
        0 \\
        f_h(x_2)
   \end{smallmatrix}
    \end{pmatrix}.
\end{equation*}

\noindent {\bf Stage $\textcircled 3$: perform $A(\eta_h)...A(\eta_1)$.}  Let $(\bar{a}_0, \bar{b}_0, \bar{c}_0, \bar{d}_0, \bar{e}_0, \bar{f}_0, \bar{g}_0, \bar{l}_0)=(0,1,0,1,0,1,0,1)$ and
\begin{equation*}
     \begin{pmatrix}
    \begin{smallmatrix}
        e^{i\beta_h}g_h(x_1)\bar{a}_k \\
        e^{i\beta_h}g_h(x_1)\bar{b}_k \\
        e^{i\beta_h}e_h(x_2)\bar{c}_k \\
        e^{i\beta_h}e_h(x_2)\bar{d}_k \\
        g_h(x_1)\bar{e}_k \\
        g_h(x_1)\bar{f}_k \\
        f_h(x_2)\bar{g}_k \\
        f_h(x_2)\bar{l}_k
   \end{smallmatrix}
    \end{pmatrix}
    =A(\zeta_{k})...A(\zeta_1)
  \begin{pmatrix}
    \begin{smallmatrix}
        0 \\
        e^{i\beta_h}g_h(x_1) \\
        0 \\
        e^{i\beta_h}e_h(x_2) \\
        0 \\
        g_h(x_1) \\
        0 \\
        f_h(x_2)
   \end{smallmatrix}
    \end{pmatrix}
\end{equation*}
for $k=1,...,h$.
The recurrence formula of $\bar{d}_k(x_1)$ is defined by $\bar{d}_0(x_1)=1,\bar{d}_1(x_1)=x_1$ and for $k=2,...,h$
\begin{equation*}
    \bar{d}_k(x_1)=x(1+e^{-i(\eta_k-\eta_{k-1})})\bar{d}_{k-1}(x_1)-e^{-i(\eta_k-\eta_{k-1})}\bar{d}_{k-2}(x_1),
\end{equation*}
with $x_1=\cos(\frac{\omega_1}{2})$, and
the recurrence formula of $\bar{l}_k(x_1)$ is defined by $\bar{l}_0(x_1)=1,\bar{l}_1(x_1)=x_1$ and for $k=2,...,h$
\begin{equation*}
    \bar{l}_k(x_1)=x_1(1+e^{-i(\eta_k-\eta_{k-1})})\bar{l}_{k-1}(x_1)-e^{-i(\eta_k-\eta_{k-1})}\bar{l}_{k-2}(x_1).
\end{equation*}
The recurrence formula of $\bar{b}_k(x_2)$ is defined by $\bar{b}_0(x_2)=1,\bar{b}_1(x_2)=x_2$ and for $k=2,...,h$
\begin{equation*}
    \bar{b}_k(x_2)=x_2(1+e^{-i(\eta_k-\eta_{k-1})})\bar{b}_{k-1}(x_2)-e^{-i(\eta_k-\eta_{k-1})}\bar{b}_{k-2}(x_2)
\end{equation*}
with $x_2=\cos(\frac{\omega_2}{2})$, and
the recurrence formula of $\bar{f}_k(x_2)$ is defined by $\bar{f}_0(x_2)=1,\bar{f}_1(x_2)=x_2$ and for $k=2,...,h$
\begin{equation*}
    \bar{f}_k(x_2)=x_2(1+e^{-i(\eta_k-\eta_{k-1})})\bar{f}_{k-1}(x_2)-e^{-i(\eta_k-\eta_{k-1})}\bar{f}_{k-2}(x_2).
\end{equation*}
By Lemma \ref{lemma1}, when
\begin{equation}\label{twoSideOddA4X}
  \eta_{k+1}-\eta_{k}=(-1)^k\pi-2\arccot(\tan(\frac{k\pi}{h})\sqrt{1-\gamma^{2}}),
\end{equation}
with $\gamma^{-1}=\cos(\frac{1}{h}\arccos(\frac{1}{\sqrt{\epsilon}}))$ for $k=1,2,\dots,h-1,$ we have
\begin{equation*}
 \bar{d}_h(x_1)=\bar{l}_h(x_1)=\frac{T_h(\frac{x_1}{\gamma})}{T_h(\frac{1}{\gamma})}, \quad \bar{b}_h(x_2)=\bar{f}_h(x_2)=\frac{T_h(\frac{x_2}{\gamma})}{T_h(\frac{1}{\gamma})},
\end{equation*}
with $T_h(\frac{1}{\gamma})=\frac{1}{\sqrt{\epsilon}}$.
Moreover, $\bar{a}_h(x_2)$, $\bar{c}_h(x_1)$, $\bar{e}_h(x_2)$, and $\bar{g}_h(x_1)$ are determined by $|\bar{a}_h(x_2)|^2+|\bar{b}_h(x_2)|^2=1$, $|\bar{c}_h(x_1)|^2+|\bar{c}_h(x_1)|^2=1$, $|\bar{e}_h(x_2)|^2+|\bar{e}_h(x_2)|^2=1$, and $|\bar{g}_h(x_1)|^2+|\bar{l}_h(x_1)|^2=1$, respectively.
Hence, the result state is
\begin{equation*}
  \frac{1}{\sqrt{2}} \begin{pmatrix}
    \begin{smallmatrix}
        e^{i\beta_h}g_h(x_1)\bar{a}_h(x_2) \\
        e^{i\beta_h}g_h(x_1)\bar{b}_h(x_2) \\
        e^{i\beta_h}e_h(x_2)\bar{c}_h(x_1) \\
        e^{i\beta_h}e_h(x_2)\bar{d}_h(x_1) \\
        g_h(x_1)\bar{e}_h(x_2) \\
        g_h(x_1)\bar{f}_h(x_2) \\
        f_h(x_2)\bar{g}_h(x_1) \\
        f_h(x_2)\bar{l}_h(x_1)
   \end{smallmatrix}
    \end{pmatrix}.
\end{equation*}

\noindent {\bf Stage $\textcircled 4$: perform  the final operation $S$.}
The  final state is
\begin{equation*}
  \frac{1}{\sqrt{2}} \begin{pmatrix}
    \begin{smallmatrix}
         e^{i\beta_h}e_h(x_2)\bar{c}_h(x_1)\\
        f_h(x_2)\bar{g}_h(x_1) \\
        e^{i\beta_h}g_h(x_1)\bar{a}_h(x_2) \\
         g_h(x_1)\bar{e}_h(x_2)\\
       e^{i\beta_h}e_h(x_2)\bar{d}_h(x_1)  \\
       f_h(x_2)\bar{l}_h(x_1) \\
        e^{i\beta_h}g_h(x_1)\bar{b}_h(x_2) \\
         g_h(x_1)\bar{f}_h(x_2)
   \end{smallmatrix}
    \end{pmatrix}.
\end{equation*}

Hence, the success probability $P_h$ is \begin{align*}
    P_h &=1-\frac{1}{2}|f_h(x_2)\bar{l}_h(x_1)|^2-\frac{1}{2}|g_h(x_1)\bar{f}_h(x_2)|^2
        = 1-\epsilon^2 T_h^2(\frac{x_1}{\gamma})T_h^2(\frac{x_2}{\gamma})\\ \nonumber
        &=1-\epsilon^2 T_h^2(\cos(\frac{1}{h}\arccos(\frac{1}{\sqrt{\epsilon}}))\sqrt{1-\frac{n_l}{N_l}}))T_h^2(\cos(\frac{1}{h}\arccos(\frac{1}{\sqrt{\epsilon}}))\sqrt{1-\frac{n_r}{N_r}})).
\end{align*}

By Eq.~\eqref{twoSideOddA1}  and Eqs.~\eqref{twoSetOddEtaAlpha}-\eqref{twoSideOddA4X}, $\alpha_k, \beta_k$ can be chosen such that
$$\alpha_k=-\beta_{h+2-k}=\pi+(\eta_{k+1}-\eta_{k})=2\arccot(\tan(\frac{k\pi}{h})\sqrt{1-\gamma^2})$$
for $k=2,4,\dots,h-1$, and
$$\alpha_k=-\beta_{h-k}=\pi-(\zeta_{k+1}-\zeta_{k})=\\2\arccot(\tan(\frac{(k-1)\pi}{h})\sqrt{1-\gamma^2}$$ for $k=3,5,\dots,h$. In addition, $\alpha_1$ and $\beta_h$ can be any value.

\subsection{Case 2: $h$ is an even integer}\label{subsection5.2}

First we set
\begin{equation}\label{twoSideEvenA1}
    \beta_k=
    -\alpha_{h+1-k} \text{ \quad \quad $k=1,2,\dots,h-1$}.
\end{equation}
Then $\ket{\Psi_h}$ reduces to
\begin{equation}\label{twoSideEvenReduce}
\begin{split}
    \ket{\Psi_h} \sim R(\beta_h)\big[A(\phi_{h-1})...A(\phi_1)\big]R(\alpha_1)S\big[A(\psi_{h+1})...A(\psi_1)\big]\ket{\bar{0}},
\end{split}
\end{equation}
which will be proven in Appendix \ref{AppendixD} by using Eqs.~\eqref{C11}-\eqref{C16}.
Here $\phi_{k}=\phi_{h-k}$ for $k=1,2,\dots,h-1$, $\psi_{k}=\psi_{h+2-k}$
for $k=1,2,\dots,h+1,$
and
\begin{equation}\label{twoSideEvenA2}
    \phi_{k+1}-\phi_{k}=
    \begin{cases}
    \pi - \alpha_{k+1} & \text{ $k=2,4,\dots,h-2$},\\
    -\pi+\alpha_{h-k} & \text{ $k=1,3,\dots,h-3$},
    \end{cases}
\end{equation}
\begin{equation}\label{twoSideEvenA3}
    \psi_{k+1}-\psi_{k}=
    \begin{cases}
    \pi - \alpha_{k} & \text{ $k=2,4,\dots,h-2$},\\
    -\pi+\alpha_{h-k+1} & \text{ $k=1,3,\dots,h$}.
    \end{cases}
\end{equation}

The final state of Eq.~\eqref{twoSideEvenReduce}  is given by the four stages as follows:
\begin{align*}
 \frac{1}{\sqrt{2}}
 \begin{pmatrix}
    \begin{smallmatrix}
        0 \\
        0 \\
        0 \\
        0 \\
        0 \\
       1 \\
        0 \\
        1 \\
   \end{smallmatrix}
    \end{pmatrix}
&\xrightarrow [\textcircled 1] {A(\psi_{h+1})...A(\psi_1)}
  \frac{1}{\sqrt{2}}
 \begin{pmatrix}
    \begin{smallmatrix}
        0 \\
        0 \\
        0 \\
        0 \\
        e_{h+1}(x_2)\\
      f_{h+1}(x_2) \\
        g_{h+1}(x_1) \\
        l_{h+1}(x_1)
     \end{smallmatrix}
    \end{pmatrix}
\xrightarrow [\textcircled 2] {R(\alpha_1)S}
 \frac{1}{\sqrt{2}} \begin{pmatrix}
    \begin{smallmatrix}
        0 \\
        g_{h+1}(x_1) \\
        0 \\
        e_{h+1}(x_2) \\
        0 \\
        l_{h+1}(x_1) \\
        0 \\
        f_{h+1}(x_2)
   \end{smallmatrix}
    \end{pmatrix}
\\&\xrightarrow [\textcircled 3] {A(\phi_{h-1})...A(\phi_1)}   \frac{1}{\sqrt{2}} \begin{pmatrix}
    \begin{smallmatrix}
         g_{h+1}(x_1)\bar{a}_{h-1}(x_2) \\
         g_{h+1}(x_1)\bar{b}_{h-1}(x_2)  \\
        e_{h+1}(x_2)\bar{c}_{h-1}(x_1) \\
        e_{h+1}(x_2)\bar{d}_{h-1}(x_1)  \\
        l_{h+1}(x_1)\bar{e}_{h-1}(x_2) \\
        l_{h+1}(x_1)\bar{f}_{h-1}(x_2) \\
        f_{h+1}(x_2)\bar{g}_{h-1}(x_1) \\
        f_{h+1}(x_2)\bar{l}_{h-1}(x_1)
   \end{smallmatrix}
    \end{pmatrix}
    \\
&\xrightarrow [\textcircled 4] {\quad \quad R(\beta_h) \quad \quad}  \frac{1}{\sqrt{2}} \begin{pmatrix}
    \begin{smallmatrix}
        e^{i\beta_h} g_{h+1}(x_1)\bar{a}_{h-1}(x_2) \\
         g_{h+1}(x_1)\bar{b}_{h-1}(x_2)  \\
        e^{i\beta_h}e_{h+1}(x_2)\bar{c}_{h-1}(x_1) \\
        e_{h+1}(x_2)\bar{d}_{h-1}(x_1)  \\
        e^{i\beta_h}l_{h+1}(x_1)\bar{e}_{h-1}(x_2) \\
        l_{h+1}(x_1)\bar{f}_{h-1}(x_2) \\
       e^{i\beta_h} f_{h+1}(x_2)\bar{g}_{h-1}(x_1) \\
        f_{h+1}(x_2)\bar{l}_{h-1}(x_1)
   \end{smallmatrix}
    \end{pmatrix}.
\end{align*}

\noindent {\bf Stage $\textcircled 1$:  apply $A(\psi_{h+1})...A(\psi_1)$  to the initial state.} Let $$(a_0, b_0, c_0, d_0, e_0, f_0, g_0, l_0)=(0,0,0,0,0,1,0,1)$$ and
\begin{equation*}
 \begin{pmatrix}
    \begin{smallmatrix}
        a_k \\
        b_k \\
        c_k \\
        d_k \\
        e_k \\
      f_k \\
       g_k \\
        l_k
     \end{smallmatrix}
    \end{pmatrix}
    =A(\psi_{k})...A(\psi_1)
 \begin{pmatrix}
    \begin{smallmatrix}
        0 \\
        0 \\
        0 \\
        0 \\
        0 \\
      1\\
        0 \\
        1
     \end{smallmatrix}
    \end{pmatrix}
\end{equation*}
for $k=1,2,\dots,h.$ By the property of the matrix $A(\zeta_{i})$, we have
$a_k=b_k=c_k=d_k=0$ for all $k.$ And $e_k, f_k, g_k,l_k$ can be gotten in the following.
The recurrence formula of $l_k(x_1)$ is defined by $l_0(x_1)=1,l_1(x_1)=x_1$ and for $k=2,...,h+1$
\begin{equation*}
    l_k(x_1)=x_1(1+e^{-i(\psi_k-\psi_{k-1})})l_{k-1}(x_1)-e^{-i(\psi_k-\psi_{k-1})}l_{k-2}(x_1)
\end{equation*}
with $x_1=\cos(\frac{\omega_1}{2})$, and
the recurrence formula of $f_k(x_2)$ is defined by $f_0(x_2)=1,f_1(x_2)=x_2$ and for $k=2,...,h+1$
\begin{equation*}
    f_k(x_2)=x_2(1+e^{-i(\psi_k-\psi_{k-1})})f_{k-1}(x_2)-e^{-i(\psi_k-\psi_{k-1})}f_{k-2}(x_2)
\end{equation*} with $x_2=\cos(\frac{\omega_2}{2})$.
By Lemma \ref{lemma1}, when
\begin{equation}\label{twoSideEvenA4}
    \psi_{k+1}-\psi_{k}=(-1)^k\pi-2\arccot(\tan(\frac{k\pi}{h+1})\sqrt{1-\gamma_1^{2}})
 \end{equation}
 with $\gamma^{-1}_1=\cos(\frac{1}{h+1}\arccos(\frac{1}{\sqrt{\epsilon}}))$ for $k=1,2,\dots,h,$
we have
\begin{equation*}
    l_{h+1}(x_1)=\frac{T_{h+1}(\frac{x_1}{\gamma_1})}{T_{h+1}(\frac{1}{\gamma_1})},\quad  f_{h+1}(x_2)=\frac{T_{h+1}(\frac{x_2}{\gamma_1})}{T_{h+1}(\frac{1}{\gamma_1})},
\end{equation*}
with $T_{h+1}(\frac{1}{\gamma_1})=\frac{1}{\sqrt{\epsilon}}$.
Moreover, $e_{h+1}(x_2)$ and $g_{h+1}(x_1)$ are determined by $|e_{h+1}(x_2)|^2+|f_{h+1}(x_2)|^2=1$ and $|g_{h+1}(x_1)|^2+|l_{h+1}(x_1)|^2=1$, respectively. Therefore, the state after $A(\psi_{h})...A(\psi_1)$ applied to the initial state is
\begin{equation*}
 \frac{1}{\sqrt{2}}
 \begin{pmatrix}
    \begin{smallmatrix}
        0 \\
        0 \\
        0 \\
        0 \\
        e_{h+1}(x_2)\\
      f_{h+1}(x_2) \\
        g_{h+1}(x_1) \\
        l_{h+1}(x_1)
     \end{smallmatrix}
    \end{pmatrix}.
\end{equation*}

\noindent {\bf Stage $\textcircled 2$: apply $R(\alpha_1)S$ to the above state.} After that, the state is
\begin{equation*}
    \frac{1}{\sqrt{2}} \begin{pmatrix}
    \begin{smallmatrix}
        0 \\
        g_{h+1}(x_1) \\
        0 \\
        e_{h+1}(x_2) \\
        0 \\
        l_{h+1}(x_1) \\
        0 \\
        f_{h+1}(x_2)
   \end{smallmatrix}
    \end{pmatrix}.
\end{equation*}

\noindent {\bf Stage $\textcircled 3$: perform $A(\phi_{h-1})...A(\phi_1)$.}  Let $(\bar{a}_0, \bar{b}_0, \bar{c}_0, \bar{d}_0, \bar{e}_0, \bar{f}_0, \bar{g}_0, \bar{l}_0)=(0,1,0,1,0,1,0,1)$ and
\begin{equation*}
    \frac{1}{\sqrt{2}} \begin{pmatrix}
    \begin{smallmatrix}
         g_{h+1}(x_1)\bar{a}_k \\
         g_{h+1}(x_1)\bar{b}_k \\
        e_{h+1}(x_2)\bar{c}_k \\
        e_{h+1}(x_2)\bar{d}_k \\
        l_{h+1}(x_1)\bar{e}_k \\
        l_{h+1}(x_1)\bar{f}_k \\
        f_{h+1}(x_2)\bar{g}_k \\
        f_{h+1}(x_2)\bar{l}_k
   \end{smallmatrix}
    \end{pmatrix}
    =A(\phi_{k})...A(\phi_1)
   \frac{1}{\sqrt{2}} \begin{pmatrix}
    \begin{smallmatrix}
        0 \\
        g_{h+1}(x_1) \\
        0 \\
        e_{h+1}(x_2) \\
        0 \\
        l_{h+1}(x_1) \\
        0 \\
        f_{h+1}(x_2)
   \end{smallmatrix}
    \end{pmatrix}
\end{equation*}
for $k=1,...,h-1$.
The recurrence formula of $\bar{d}_k(x_1)$ is defined by $\bar{d}_0(x_1)=1,\bar{d}_1(x_1)=x_1$ and for $k=2,...,h-1$
\begin{equation*}
    \bar{d}_k(x_1)=x(1+e^{-i(\phi_k-\phi_{k-1})})\bar{d}_{k-1}(x_1)-e^{-i(\phi_k-\phi_{k-1})}\bar{d}_{k-2}(x_1)
\end{equation*} with $x_1=\cos(\frac{\omega_1}{2})$,
and
the recurrence formula of $\bar{l}_k(x_1)$ is defined by $\bar{l}_0(x_1)=1,\bar{l}_1(x_1)=x_1$ and for $k=2,...,h-1$
\begin{equation*}
    \bar{l}_k(x_1)=x_1(1+e^{-i(\phi_k-\phi_{k-1})})\bar{l}_{k-1}(x_1)-e^{-i(\phi_k-\phi_{k-1})}\bar{l}_{k-2}(x_1).
\end{equation*}

The recurrence formula of $\bar{b}_k(x_2)$ is defined by $\bar{b}_0(x_2)=1,\bar{b}_1(x_2)=x_2$ and for $k=2,...,h-1$
\begin{equation*}
    \bar{b}_k(x_2)=x_2(1+e^{-i(\phi_k-\phi_{k-1})})\bar{b}_{k-1}(x_2)-e^{-i(\phi_k-\phi_{k-1})}\bar{b}_{k-2}(x_2)
\end{equation*}
 with  $x_2=\cos(\frac{\omega_2}{2})$, and
the recurrence formula of $\bar{f}_k(x_2)$ is defined by $\bar{f}_0(x_2)=1,\bar{f}_1(x_2)=x_2$ and for $k=2,...,h-1$
\begin{equation*}
    \bar{f}_k(x_2)=x_2(1+e^{-i(\phi_k-\phi_{k-1})})\bar{f}_{k-1}(x_2)-e^{-i(\phi_k-\phi_{k-1})}\bar{f}_{k-2}(x_2).
\end{equation*}
By Lemma \ref{lemma1},
when
\begin{equation}\label{twoSideEvenA5}
    \phi_{k+1}-\phi_{k}=(-1)^k\pi-2\arccot(\tan(\frac{k\pi}{h-1})\sqrt{1-\gamma_2^{2}})
 \end{equation}
 with $\gamma^{-1}_2=\cos(\frac{1}{h-1}\arccos(\frac{1}{\sqrt{\epsilon}}))$
 for $k=1,2,\dots,h-2,$
we have
\begin{equation*}
 \bar{d}_{h-1}(x_1)=\bar{l}_{h-1}(x_1)=\frac{T_{h-1}(\frac{x_1}{\gamma_2})}{T_{h-1}(\frac{1}{\gamma_2})}, \quad \bar{b}_{h-1}(x_2)=\bar{f}_{h-1}(x_2)=\frac{T_{h-1}(\frac{x_2}{\gamma_2})}{T_{h-1}(\frac{1}{\gamma_2})},
\end{equation*}
with $T_{h-1}(\frac{1}{\gamma_2})=\frac{1}{\sqrt{\epsilon}}$.
Moreover, $\bar{a}_h(x_2)$, $\bar{c}_h(x_1)$, $\bar{e}_h(x_2)$, and $\bar{g}_h(x_1)$ are determined by $|\bar{a}_h(x_2)|^2+|\bar{b}_h(x_2)|^2=1$, $|\bar{c}_h(x_1)|^2+|\bar{d}_h(x_1)|^2=1$, $|\bar{e}_h(x_2)|^2+|\bar{f}_h(x_2)|^2=1$, and $|\bar{g}_h(x_1)|^2+|\bar{l}_h(x_1)|^2=1$, respectively.
Hence, the result state is
\begin{equation*}
    \frac{1}{\sqrt{2}} \begin{pmatrix}
    \begin{smallmatrix}
         g_{h+1}(x_1)\bar{a}_{h-1}(x_2) \\
         g_{h+1}(x_1)\bar{b}_{h-1}(x_2)  \\
        e_{h+1}(x_2)\bar{c}_{h-1}(x_1) \\
        e_{h+1}(x_2)\bar{d}_{h-1}(x_1)  \\
        l_{h+1}(x_1)\bar{e}_{h-1}(x_2) \\
        l_{h+1}(x_1)\bar{f}_{h-1}(x_2) \\
        f_{h+1}(x_2)\bar{g}_{h-1}(x_1) \\
        f_{h+1}(x_2)\bar{l}_{h-1}(x_1)
   \end{smallmatrix}
    \end{pmatrix}.
\end{equation*}

\noindent {\bf Stage $\textcircled 4$: perform  the final operation $R(\beta_h)$.}
The  final state is
\begin{equation*}
    \frac{1}{\sqrt{2}} \begin{pmatrix}
    \begin{smallmatrix}
        e^{i\beta_h} g_{h+1}(x_1)\bar{a}_{h-1}(x_2) \\
         g_{h+1}(x_1)\bar{b}_{h-1}(x_2)  \\
        e^{i\beta_h}e_{h+1}(x_2)\bar{c}_{h-1}(x_1) \\
        e_{h+1}(x_2)\bar{d}_{h-1}(x_1)  \\
        e^{i\beta_h}l_{h+1}(x_1)\bar{e}_{h-1}(x_2) \\
        l_{h+1}(x_1)\bar{f}_{h-1}(x_2) \\
       e^{i\beta_h} f_{h+1}(x_2)\bar{g}_{h-1}(x_1) \\
        f_{h+1}(x_2)\bar{l}_{h-1}(x_1)
   \end{smallmatrix}
    \end{pmatrix}.
\end{equation*}

Therefore, the success probability $P_{h}$ is
\begin{align*}
    P_h &= 1-\frac{1}{2}|c_{h+1}(x_1)d_{h-1}(x_2)|^2-\frac{1}{2}|c_{h+1}(x_2)d_{h-1}(x_1)|^2\\
    &=1-\frac{\epsilon^2}{2}( T^2_{h+1}(\frac{x_1}{\gamma_1}) T^2_{h-1}(\frac{x_2}{\gamma_2})
    + T^2_{h+1}(\frac{x_2}{\gamma_1}) T^2_{h-1}(\frac{x_1}{\gamma_2}))\\
    &=1-\frac{\epsilon^2}{2}\bigg[ T_{h+1}^2\left(\cos\left(\frac{1}{h+1}\arccos\left(\frac{1}{\sqrt{\epsilon}}\right)\right)\sqrt{1-\frac{n_l}{N_l}}\right) \cdot
   \\&~~~~ T_{h-1}^2\left(\cos\left(\frac{1}{h-1} \arccos\left(\frac{1}{\sqrt{\epsilon}}\right)\right)\sqrt{1-\frac{n_r}{N_r}}\right)+
    \\&~~~~ T_{h+1}^2\left(\cos(\frac{1}{h+1}\arccos(\frac{1}{\sqrt{\epsilon}}))\sqrt{1-\frac{n_r}{N_r}}\right)\cdot
    \\
    &~~~~ T_{h-1}^2\left(\cos(\frac{1}{h-1}\arccos(\frac{1}{\sqrt{\epsilon}}))\sqrt{1-\frac{n_l}{N_l}}\right)
    \bigg]
\end{align*}

By Eq.~\eqref{twoSideEvenA1} and Eqs.~\eqref{twoSideEvenA2}-\eqref{twoSideEvenA5}, $\alpha_k, \beta_k$ can be chosen such that
\begin{equation*}
    \alpha_k=-\beta_{h+1-k}=\pi -(\phi_k-\phi_{k-1})\\=2\arccot(\tan(\frac{k\pi}{h+1})\sqrt{1-\gamma_2^2})
\end{equation*}
for $k=2,4,\dots,h$, and
\begin{equation*}
    \alpha_k=-\beta_{h+1-k}=\pi -(\psi_{k+1}-\psi_k)\\=2\arccot(\tan(\frac{(k-1)\pi}{h-1})\sqrt{1-\gamma_1^2})
\end{equation*}
for $k=3,5,\dots,h-1$.
In addition, $\alpha_1$ and $\beta_h$ can be any value.

\section{Proof of Equations (\ref{oneSideOddreduced}) and (\ref{oneSideEvenreduced})}\label{AppendixC}
In this appendix, we give the detailed proof of Equations~\eqref{oneSideOddreduced} and (\ref{oneSideEvenreduced}).

Recall that
\begin{equation*}
    \ket{\Psi_h}=U(\alpha_h,\beta_h) \dots U(\alpha_1,\beta_1)\ket{\Psi_0}\\=SC(\alpha_h)Q(\beta_h)\dots SC(\alpha_1)Q(\beta_1)\ket{\Psi_0}.
\end{equation*}
Using $SS=I$, we have
\begin{equation*}
   \ket{\Psi_h}=SC(\alpha_h)Q(\beta_h)\dots SC(\alpha_1)Q(\beta_1)SS\ket{\Psi_0}.
\end{equation*}
Using Eq.~(\ref{C=ARA}) and Eq.~(\ref{QS=SR}), we have
\begin{equation*}
\begin{split}
   \ket{\Psi_h}\sim SA(\frac{\pi}{2})R(\alpha_h)A(-\frac{\pi}{2})SR(\beta_h)A(\frac{\pi}{2})R(\alpha_{h-1})A(-\frac{\pi}{2})SR(\beta_{h-1})\dots\\    A(\frac{\pi}{2})R(\alpha_2)A(-\frac{\pi}{2})SR(\beta_2)A(\frac{\pi}{2})R(\alpha_1)A(-\frac{\pi}{2})SR(\beta_1)S\ket{\Psi_0}.
\end{split}
\end{equation*}
The number of $S$ is $h+2$. Two cases need to be considered.

\textbf{Case 1: h is odd.}
By Eq. (\ref{SABSCDS=CDSAB}), we have,

\begin{align*}
    \ket{\Psi_h}\sim
    &SA(\frac{\pi}{2})R(\alpha_h)A(-\frac{\pi}{2})\dots SR(\beta_{\frac{h+5}{2}})A(\frac{\pi}{2})R(\alpha_{\frac{h+3}{2}})A(-\frac{\pi}{2})\\ \nonumber
   &\Big\{ SR(\beta_{\frac{h+3}{2}})
    A(\frac{\pi}{2})R(\alpha_{\frac{h+1}{2}})A(-\frac{\pi}{2})SR(\beta_{\frac{h+1}{2}})A(\frac{\pi}{2})R(\alpha_{\frac{h-1}{2}})A(-\frac{\pi}{2})
    S \Big\}
    \\ \nonumber &
    R(\beta_{\frac{h-1}{2}})
    A(\frac{\pi}{2})R(\alpha_{\frac{h-3}{2}})A(-\frac{\pi}{2}) \dots SR(\beta_2)A(\frac{\pi}{2})R(\alpha_1)A(-\frac{\pi}{2})SR(\beta_1)S\ket{\Psi_0}\\\nonumber
    =  &SA(\frac{\pi}{2})R(\alpha_h)A(-\frac{\pi}{2})\dots SR(\beta_{\frac{h+5}{2}})A(\frac{\pi}{2})R(\alpha_{\frac{h+3}{2}})A(-\frac{\pi}{2})\\ \nonumber
   &\Big\{ R(\beta_{\frac{h+1}{2}})A(\frac{\pi}{2})R(\alpha_{\frac{h-1}{2}})A(-\frac{\pi}{2})SR(\beta_{\frac{h+3}{2}})
    A(\frac{\pi}{2})R(\alpha_{\frac{h+1}{2}})A(-\frac{\pi}{2})
     \Big\}
    \\ \nonumber &
    R(\beta_{\frac{h-1}{2}})
    A(\frac{\pi}{2})R(\alpha_{\frac{h-3}{2}})A(-\frac{\pi}{2}) \dots SR(\beta_2)A(\frac{\pi}{2})R(\alpha_1)A(-\frac{\pi}{2})SR(\beta_1)S\ket{\Psi_0}\\\nonumber
    =&SA(\frac{\pi}{2})R(\alpha_h)A(-\frac{\pi}{2})R(\beta_{h-1})A(\frac{\pi}{2}) \dots A(\frac{\pi}{2})R(\alpha_1)A(-\frac{\pi}{2}) \\&SR(\beta_h)A(\frac{\pi}{2})R(\alpha_{h-1})A(-\frac{\pi}{2}) \dots R(\alpha_2)A(-\frac{\pi}{2})R(\beta_1)S\ket{\Psi_0}.
\end{align*}
The process is as follows. We first select the formula $$ \Big \{ SR(\beta_{\frac{h+3}{2}})
    A(\frac{\pi}{2})R(\alpha_{\frac{h+1}{2}})A(-\frac{\pi}{2})SR(\beta_{\frac{h+1}{2}})A(\frac{\pi}{2})R(\alpha_{\frac{h-1}{2}})A(-\frac{\pi}{2})
    S \Big\}$$ containing the middle $S$, and then this formula reduces to $$\Big\{ R(\beta_{\frac{h+1}{2}})A(\frac{\pi}{2})R(\alpha_{\frac{h-1}{2}})A(-\frac{\pi}{2})SR(\beta_{\frac{h+3}{2}})
    A(\frac{\pi}{2})R(\alpha_{\frac{h+1}{2}})A(-\frac{\pi}{2})
     \Big\}$$ according to  Eq. (\ref{SABSCDS=CDSAB}). The final result can be obtained after repeating the above steps. 
Using Eq. (\ref{ASA0}) and $A(-\frac{\pi}{2})A(\frac{\pi}{2})=I$, we have
\begin{align*}
     \ket{\Psi_h}\sim&
    SA(\frac{\pi}{2})R(\alpha_h)A(-\frac{\pi}{2})R(\beta_{h-1})A(\frac{\pi}{2})...A(\frac{\pi}{2})R(\alpha_1)A(-\frac{\pi}{2}) \\&SR(\beta_h)A(\frac{\pi}{2})R(\alpha_{h-1})A(-\frac{\pi}{2})...R(\alpha_2)A(-\frac{\pi}{2})R(\beta_1)SA(\frac{\pi}{2})SA(\frac{\pi}{2})\ket{\bar{0}}\\
   \sim &SA(\frac{\pi}{2})R(\alpha_h)A(-\frac{\pi}{2})R(\beta_{h-1})A(\frac{\pi}{2})...R(\beta_2)A(\frac{\pi}{2})R(\alpha_1)\\
   &SR(\beta_h)A(\frac{\pi}{2})R(\alpha_{h-1})A(-\frac{\pi}{2})...R(\alpha_2)A(-\frac{\pi}{2})R(\beta_1)A(\frac{\pi}{2})\ket{\bar{0}}.
\end{align*}

Here, two cases need to be discussed.

\textbf{Case 1.1:} When $h \bmod 4=1$, $\beta_i=\alpha_{h+2-i}$ for $i=2,4,\dots,h-1$, and $\beta_i=\alpha_{h-i}$ for $i=1,3,\dots,h-2$,
using Eqs. (\ref{A=RAR}) and (\ref{RR=I}), we have
\begin{align}\label{h141}
    \ket{\Psi_h}\sim
    &SA(\frac{\pi}{2})A(-\frac{\pi}{2}+\alpha_h)A(\frac{\pi}{2}-\alpha_3+\alpha_h)A(-\frac{\pi}{2}-\alpha_3+\alpha_h+\alpha_{h-2})\nonumber\\
    &\dots A(\frac{\pi}{2}-\alpha_3-\alpha_5\dots-\alpha_{k_1}+\alpha_{k_1+2}\dots+\alpha_h)\nonumber\\
    &\dots A(-\frac{\pi}{2}-\alpha_3+\alpha_h+\alpha_{h-2})A(\frac{\pi}{2}-\alpha_3+\alpha_h)A(-\frac{\pi}{2}+\alpha_h) A(\frac{\pi}{2})R(\alpha_1)\nonumber\\
    &SR(\beta_h)A(\frac{\pi}{2})A(-\frac{\pi}{2}+\alpha_{h-1})A(\frac{\pi}{2}-\alpha_2+\alpha_{h-1})A(-\frac{\pi}{2}-\alpha_2+\alpha_{h-1}+\alpha_{h-3})\nonumber\\
    &\dots A(\frac{\pi}{2}-\alpha_2-\alpha_4\dots-\alpha_{k_2}+\alpha_{k_2+2}\dots+\alpha_{h-1})\nonumber\\
    &\dots A(-\frac{\pi}{2}-\alpha_2+\alpha_{h-1}+\alpha_{h-3})A(\frac{\pi}{2}-\alpha_2+\alpha_{h-1})A(-\frac{\pi}{2}+\alpha_{h-1}) A(\frac{\pi}{2})\ket{\bar{0}},
\end{align}
where $k_1=(h+1)/2$, $k_2=(h-1)/2$.

\textbf{Case 1.2:} When $h\bmod 4=3$, $\beta_i=\alpha_{h+2-i}$ for $i=2,4,\dots,h-1$, and $\beta_i=\alpha_{h-i}$ for $i=1,3,\dots,h-2$,
using Eqs. (\ref{A=RAR}) and (\ref{RR=I}), we have
\begin{align}\label{h143}
     \ket{\Psi_h}\sim
    &SA(\frac{\pi}{2})A(-\frac{\pi}{2}+\alpha_h)A(\frac{\pi}{2}-\alpha_3+\alpha_h)A(-\frac{\pi}{2}-\alpha_3+\alpha_h+\alpha_{h-2})\nonumber\\
    &\dots A(-\frac{\pi}{2}-\alpha_3-\alpha_5\dots-\alpha_{k_1}+\alpha_{k_1+2}\dots+\alpha_h)\nonumber\\
    &\dots A(-\frac{\pi}{2}-\alpha_3+\alpha_h+\alpha_{h-2})A(\frac{\pi}{2}-\alpha_3+\alpha_h)A(-\frac{\pi}{2}+\alpha_h) A(\frac{\pi}{2})R(\alpha_1)\nonumber\\
    &SR(\beta_h)A(\frac{\pi}{2})A(-\frac{\pi}{2}+\alpha_{h-1})A(\frac{\pi}{2}-\alpha_2+\alpha_{h-1})A(-\frac{\pi}{2}-\alpha_2+\alpha_{h-1}+\alpha_{h-3})\nonumber
    \\&\dots A(-\frac{\pi}{2}-\alpha_2-\alpha_4\dots-\alpha_{k_2}+\alpha_{k_2+2}\dots+\alpha_{h-1})\nonumber\\
    &\dots A(-\frac{\pi}{2}-\alpha_2+\alpha_{h-1}+\alpha_{h-3})A(\frac{\pi}{2}-\alpha_2+\alpha_{h-1})A(-\frac{\pi}{2}+\alpha_{h-1}) A(\frac{\pi}{2})\ket{\bar{0}},
\end{align}
where $k_1=(h-1)/2$, $k_2=(h-3)/2$.

Thus, we can get Eq.~(\ref{oneSideOddreduced}) by rewriting
Eq.~(\ref{h141}) and Eq.~(\ref{h143}) as follows.
\begin{align*}\label{oneSideOddreducedc}
    \ket{\Psi_h}\sim
     S\big[A(\eta_h)...A(\eta_1)\big]R(\alpha_1)
   SR(\beta_h)\big[A(\zeta_h)...A(\zeta_1)\big]\ket{\bar{0}},
   \tag{\ref{oneSideOddreduced}}
\end{align*}
where $\eta_{k}=\eta_{h+1-k}$ and $\zeta_{k}=\zeta_{h+1-k}$ for $k=1,2,\dots,h$,
and
\begin{equation*}\label{oneSetOddEtaAlphaX}
    \eta_{k+1}-\eta_{k}=
    \begin{cases}
    \pi - \alpha_{k+1} & \text{ $k=2,4,\dots,h-1$},\\
    -\pi+\alpha_{h-k+1} & \text{ $k=3,5,\dots,h$},
    \end{cases}
\end{equation*}
\begin{equation*}\label{oneSetOddZetaAlphaX}
    \zeta_{k+1}-\zeta_{k}=
    \begin{cases}
    \pi - \alpha_{k} & \text{ $k=2,4,\dots,h-1$},\\
    -\pi+\alpha_{h-k} & \text{ $k=3,5,\dots,h$}.
    \end{cases}
\end{equation*}
Here, when $h\bmod 4=1$,
\begin{align*}
    \eta_{(h+1)/2}=&\frac{\pi}{2}+(\alpha_{(h+5)/2}+\alpha_{(h+9)/2}+...+\alpha_{h-2}+\alpha_{h})\\
    & -(\alpha_3+\alpha_5+...+\alpha_{(h-3)/2}+\alpha_{(h+1)/2}),
\end{align*}
\begin{align*}
\zeta_{(h+1)/2}=&\frac{\pi}{2}+(\alpha_{(h+3)/2}+\alpha_{(h+7)/2}+...+\alpha_{h-3}+\alpha_{h-1})\\&-(\alpha_2+\alpha_4+...+\alpha_{(h-5)/2}+\alpha_{(h-1)/2}),
\end{align*}
and when $h\bmod 4=3$,
\begin{align*}
    \eta_{(h+1)/2}=&-\frac{\pi}{2}+(\alpha_{(h+3)/2}+\alpha_{(h+7)/2}+...+\alpha_{h-2}+\alpha_{h})\\&-(\alpha_3+\alpha_5+...+\alpha_{(h-5)/2}+\alpha_{(h-1)/2}),
\end{align*}
\begin{align*}
\zeta_{(h+1)/2}=&-\frac{\pi}{2}+(\alpha_{(h+1)/2}+\alpha_{(h+5)/2}+...+\alpha_{h-3}+\alpha_{h-1})\\&-(\alpha_2+\alpha_4+...+\alpha_{(h-7)/2}+\alpha_{(h-3)/2}).
\end{align*}

\textbf{Case 2: h is even.} Using Eq. (\ref{SABSCDS=CDSAB}), there is
\begin{align*}
    \ket{\Psi_h}\sim
    &SA(\frac{\pi}{2})R(\alpha_h)A(-\frac{\pi}{2})...SR(\beta_{\frac{h}{2}+2})A(\frac{\pi}{2})R(\alpha_{\frac{h}{2}+1})A(-\frac{\pi}{2})\\
    &SR(\beta_{\frac{h}{2}+1})A(\frac{\pi}{2})R(\alpha_{\frac{h}{2}})A(-\frac{\pi}{2})SR(\beta_{\frac{h}{2}})A(\frac{\pi}{2})R(\alpha_{\frac{h}{2}-1})A(-\frac{\pi}{2})\\
    &SR(\beta_{\frac{h}{2}-1})A(\frac{\pi}{2})R(\alpha_{\frac{h}{2}-2})A(-\frac{\pi}{2})...SR(\beta_2)A(\frac{\pi}{2})R(\alpha_1)A(-\frac{\pi}{2})SR(\beta_1)S\ket{\Psi_0}\\
    =&SA(\frac{\pi}{2})R(\alpha_h)A(-\frac{\pi}{2})SR(\beta_h)A(\frac{\pi}{2})R(\alpha_{h-1})A(-\frac{\pi}{2})...\\
    &SR(\beta_{\frac{h}{2}+2})A(\frac{\pi}{2})R(\alpha_{\frac{h}{2}+1})A(-\frac{\pi}{2})R(\beta_{\frac{h}{2}})A(\frac{\pi}{2})R(\alpha_{\frac{h}{2}-1})A(-\frac{\pi}{2})\\
    &SR(\beta_{\frac{h}{2}+1})A(\frac{\pi}{2})R(\alpha_{\frac{h}{2}})A(-\frac{\pi}{2})R(\beta_{\frac{h}{2}-1})A(\frac{\pi}{2})R(\alpha_{\frac{h}{2}-2})A(-\frac{\pi}{2})S...\\
    &SR(\beta_2)A(\frac{\pi}{2})R(\alpha_1)A(-\frac{\pi}{2})SR(\beta_1)S\ket{\Psi_0}\\
   = &SA(\frac{\pi}{2})R(\alpha_h)A(-\frac{\pi}{2})R(\beta_{h-1})A(\frac{\pi}{2})...A(\frac{\pi}{2})R(\alpha_2)A(-\frac{\pi}{2})R(\beta_1) \\
    &SR(\beta_h)A(\frac{\pi}{2})R(\alpha_{h-1})A(-\frac{\pi}{2})...R(\alpha_1)A(-\frac{\pi}{2})\ket{\Psi_0}.
\end{align*}
Using Eq. (\ref{ASA0}) and $A(-\frac{\pi}{2})A(\frac{\pi}{2})=I$, we have
\begin{align*}
    \ket{\Psi_h}\sim
   & SA(\frac{\pi}{2})R(\alpha_h)A(-\frac{\pi}{2})R(\beta_{h-1})A(\frac{\pi}{2})...A(\frac{\pi}{2})R(\alpha_2)A(-\frac{\pi}{2})R(\beta_1) \\
   &SR(\beta_h)A(\frac{\pi}{2})R(\alpha_{h-1})A(-\frac{\pi}{2})...R(\alpha_1)A(-\frac{\pi}{2})A(\frac{\pi}{2})SA(\frac{\pi}{2})\ket{\bar{0}}\\
   \sim & R(\beta_h)A(\frac{\pi}{2})R(\alpha_{h-1})A(-\frac{\pi}{2})...A(-\frac{\pi}{2})R(\beta_2)A(\frac{\pi}{2})R(\alpha_1)S\\
  & A(\frac{\pi}{2})R(\alpha_h)A(-\frac{\pi}{2})R(\beta_{h-1})A(\frac{\pi}{2})...A(\frac{\pi}{2})R(\alpha_2)A(-\frac{\pi}{2})R(\beta_1)A(\frac{\pi}{2})\ket{\bar{0}}.
\end{align*}

Here, two cases need to be discussed.

\textbf{Case 2.1:}
When $h \bmod 4=0$, and $\beta_i=\alpha_{h+1-i}$ for $i=1,2,\dots,h-1$, using Eqs. (\ref{A=RAR}) and (\ref{RR=I}), we have
\begin{align}\label{h41}
 \ket{\Psi_h}\sim
    &R(\beta_h)A(\frac{\pi}{2})A(-\frac{\pi}{2}+\alpha_{h-1})A(\frac{\pi}{2}-\alpha_3+\alpha_{h-1})A(-\frac{\pi}{2}-\alpha_3+\alpha_{h-1}+\alpha_{h-3})\nonumber\\
    &...A(-\frac{\pi}{2}-\alpha_3-\alpha_5\dots-\alpha_{k_1}+\alpha_{k_1+2}\dots+\alpha_{h-1})...\nonumber\\
   & A(-\frac{\pi}{2}-\alpha_3+\alpha_{h-1}+\alpha_{h-3})A(\frac{\pi}{2}-\alpha_3+\alpha_{h-1})A(-\frac{\pi}{2}+\alpha_{h-1})A(\frac{\pi}{2})R(\alpha_1)S\nonumber\\
   & A(\frac{\pi}{2})A(-\frac{\pi}{2}+\alpha_h)A(\frac{\pi}{2}-\alpha_2+\alpha_h)A(-\frac{\pi}{2}-\alpha_2+\alpha_h+\alpha_{h-2})\nonumber\\
    &...A(\frac{\pi}{2}-\alpha_2-\alpha_4\dots-\alpha_{k_2}+\alpha_{k_2+2}\dots+\alpha_{h})...\nonumber\\
    &A(-\frac{\pi}{2}-\alpha_2+\alpha_h+\alpha_{h-2})A(\frac{\pi}{2}-\alpha_2+\alpha_h)A(-\frac{\pi}{2}+\alpha_h)A(\frac{\pi}{2})\ket{\bar{0}},
\end{align}
where $k_1=h/2-1$, $k_2=h/2$.

\textbf{Case 2.2:}
When $h \bmod 4=2$, and $\beta_i=\alpha_{h+1-i}$ for $i=1,2,\dots,h-1$, using Eqs. (\ref{A=RAR}) and (\ref{RR=I}), we have
\begin{align}\label{h43}
    \ket{\Psi_h}\sim
   & R(\beta_h)A(\frac{\pi}{2})A(-\frac{\pi}{2}+\alpha_{h-1})A(\frac{\pi}{2}-\alpha_3+\alpha_{h-1})A(-\frac{\pi}{2}-\alpha_3+\alpha_{h-1}+\alpha_{h-3})\nonumber\\
    &...A(\frac{\pi}{2}-\alpha_3-\alpha_5\dots-\alpha_{k_1}+\alpha_{k_1+2}\dots+\alpha_{h-1})...\nonumber\\
    &A(-\frac{\pi}{2}-\alpha_3+\alpha_{h-1}+\alpha_{h-3})A(\frac{\pi}{2}-\alpha_3+\alpha_{h-1})A(-\frac{\pi}{2}+\alpha_{h-1})A(\frac{\pi}{2})R(\alpha_1)S\nonumber\\
   & A(\frac{\pi}{2})A(-\frac{\pi}{2}+\alpha_h)A(\frac{\pi}{2}-\alpha_2+\alpha_h)A(-\frac{\pi}{2}-\alpha_2+\alpha_h+\alpha_{h-2})\nonumber\\
   & ...A(-\frac{\pi}{2}-\alpha_2-\alpha_4\dots-\alpha_{k_2}+\alpha_{k_2+2}\dots+\alpha_{h})...\nonumber\\
    &A(-\frac{\pi}{2}-\alpha_2+\alpha_h+\alpha_{h-2})A(\frac{\pi}{2}-\alpha_2+\alpha_h)A(-\frac{\pi}{2}+\alpha_h)A(\frac{\pi}{2})\ket{\bar{0}},
\end{align}
where $k_1=h/2$, $k_2=h/2-1$.

Using Eqs.~(\ref{h41}) and ~(\ref{h43})
we can get Eq.~(\ref{oneSideEvenreduced}),
\begin{align*}\label{oneSideEvenreducedc}
     \ket{\Psi_h} \sim &R(\beta_h)\big[A(\phi_{h-1})...A(\phi_1)\big]
 R(\alpha_1)S\big[A(\psi_{h+1})...A(\psi_1)\big]\ket{\bar{0}} \tag{\ref{oneSideEvenreduced}},
\end{align*}
where $\phi_{k}=\phi_{h-k}$ for $k=1,2,\dots,h-1$, $\psi_{k}=\psi_{h+2-k}$
for $k=1,2,\dots,h+1$,
and
\begin{equation*}
    \phi_{k+1}-\phi_{k}=
    \begin{cases}
    \pi - \alpha_{k+1} & \text{for $k=2,4,\dots,h-2$},\\
    -\pi+\alpha_{h-k} & \text{for $k=1,3,\dots,h-3$},
    \end{cases}
\end{equation*}
\begin{equation*}
    \psi_{k+1}-\psi_{k}=
    \begin{cases}
    \pi - \alpha_{k} & \text{for $k=2,4,\dots,h-2$},\\
    -\pi+\alpha_{h-k+1} & \text{for $k=1,3,\dots,h$}.
    \end{cases}
\end{equation*}
Here,
when $h \bmod 4=0$, there are
\begin{align*}
  &\phi_{h/2}=-\frac{\pi}{2}+(\alpha_{h/2+1}+\alpha_{h/2+3}+...+\alpha_{h-3}+\alpha_{h-1})-(\alpha_3+\alpha_5+...+\alpha_{h/2-3}+\alpha_{h/2-1}),\\
  &\psi_{h/2+1}=\frac{\pi}{2}+(\alpha_{h/2+2}+\alpha_{h/2+4}+...+\alpha_{h-2}+\alpha_{h})-(\alpha_2+\alpha_4+...+\alpha_{h/2-2}+\alpha_{h/2},
\end{align*}
and when $h\bmod 4=2$, there are
\begin{align*}
    &\phi_{h/2}=\frac{\pi}{2}+(\alpha_{h/2+2}+\alpha_{h/2+4}+...+\alpha_{h-3}+\alpha_{h-1})-(\alpha_3+\alpha_5+...+\alpha_{h/2-2}+\alpha_{h/2}),\\
    &\psi_{h/2+1}=-\frac{\pi}{2}+(\alpha_{h/2+1}+\alpha_{h/2+3}+...+\alpha_{h-2}+\alpha_{h})-(\alpha_2+\alpha_4+...+\alpha_{h/2-3}+\alpha_{h/2-1}).
\end{align*}

\section{Proof of Equations  \eqref{twoSideOddReduce} and \eqref{twoSideEvenReduce}}\label{AppendixD}
The process of proving Eq. \eqref{oneSideOddreduced} (\eqref{oneSideEvenreduced}) uses Eqs.~\eqref{C=ARA}-\eqref{SABSCDS=CDSAB} in Appendix \ref{AppendixC}.
There are the same mathematical equations in two cases (see Eqs.~\eqref{C=ARA}-\eqref{SABSCDS=CDSAB} and Eqs.~\eqref{C11}-\eqref{C16}), although their subspace dimensions in the analysis are different. Therefore, the proof of Eq. \eqref{twoSideOddReduce} (Eq. \eqref{twoSideEvenReduce}) is the same as Eq. \eqref{oneSideOddreduced} (\eqref{oneSideEvenreduced}).

\end{widetext}
\end{document}